\documentclass[11pt]{article}    

\usepackage{algorithm,algorithmic}
\usepackage{graphics, graphicx}
\usepackage{enumerate}
\usepackage{url}
\usepackage{mathptmx}	
\usepackage{graphicx}
\usepackage{subfigure} 
\usepackage{amsmath, amsfonts, amssymb, amsthm}
\usepackage{mathtools}

\usepackage[font=footnotesize,width=.85\textwidth,labelfont=bf]{caption}
\usepackage[mathcal]{euscript}
\usepackage{mathptmx} 

\usepackage{authblk}

\newcommand{\bone}{\mathbb{B}}
\newcommand{\Scond}{\emph{S1-condition}}
\newcommand{\skel}{\mathbb{S}}
\newcommand{\PC}{product coloring}
\newcommand{\cc}{color-continuation}
\newcommand{\la}{\langle}
\newcommand{\ra}{\rangle}

\newtheorem{thm}{Theorem}[section]
\newtheorem{cor}[thm]{Corollary}
\newtheorem{lem}[thm]{Lemma}
\newtheorem{defn}[thm]{Definition}
\newtheorem{rem}{Remark}

\begin{document}

\sloppy

\title{A Local Prime Factor Decomposition Algorithm for Strong Product Graphs}

\author[1,2,3]{Marc Hellmuth}

\affil[1]{Dpt.\ of Mathematics and Computer Science, University of Greifswald, Walther-
  Rathenau-Strasse 47, D-17487 Greifswald, Germany }
\affil[2]{Saarland University, Center for Bioinformatics, Building E 2.1, P.O.\ Box 151150, D-66041 Saarbr{\"u}cken, Germany }
\affil[3]{Bioinformatics Group, Department of Computer Science; and
		    Interdisciplinary Center of Bioinformatics, University of Leipzig, \\
			 H{\"a}rtelstra{\ss}e 16-18, D-04107 Leipzig, Germany}
\date{}

\maketitle

\abstract{ \noindent
This work is concerned with the prime factor decomposition (PFD) of 
\emph{strong product} graphs.   
 A new quasi-linear time algorithm for the PFD 
with respect to the strong product
for arbitrary, finite, connected, undirected graphs is derived.

Moreover, since most graphs are prime although they can have 
a product-like structure, also known as \emph{approximate} graph products, 
the practical application of the well-known "classical" prime factorization algorithm 
 is strictly limited.
This new PFD algorithm is based on a local approach that 
covers a graph by small factorizable subgraphs 
and then utilizes this information to derive the global factors.  
Therefore, we can take advantage of this approach and derive in addition a method   
for the recognition of approximate graph products.
}

\section{Introduction}
Graphs and in particular graph products arise in a variety of different contexts, 
from computer science \cite{AMA-07, HJH+10} to theoretical biology 
\cite{HMM-09,Wagner:03a}, 
computational engineering  \cite{KK-08, KR-04} 
or just as natural structures in discrete mathematics 
\cite{Ham-09, OHK+09,HT-14,HOS12a,HON-14}.
Standard references with respect to graph products are due to 
Imrich, Klav\v{z}ar, Douglas and Hammack \cite{IMKL-00,IMKLDO-08, IKH-11}.

In this contribution we are concerned with the \emph{prime factor decomposition}, \emph{PFD} for short,
of  \emph{strong} product graphs.
The PFD  with respect to the strong product is unique for all finite connected graphs, 
\cite{DOeIM-69, McK-71}. The first who provided 
a polynomial-time algorithm for the PFD of strong product graphs were
Feigenbaum and Sch{\"a}ffer \cite{FESC-92}. 
The latest and fastest approach is 
due to Hammack and Imrich  \cite{HAIM-09}. 
In both approaches, the key idea for the PFD of a strong product graph $G$
is to find a subgraph $\skel(G)$ of $G$ with special properties, the so-called \emph{Cartesian skeleton}, 
that is then decomposed with respect to the \emph{Cartesian} product.
Afterwards, one constructs the prime factors of $G$
using the information of the PFD of $\skel(G)$.

However, an often appearing problem can be formulated as follows:
For a given graph $G$ that has a product-like structure, 
 the task is to find a graph $H$ that is a nontrivial
product and a good approximation of $G$, in the sense that $H$
can be reached from $G$ by a small number of additions or deletions
of edges and vertices. The graph $G$ is also 
called \emph{approximate} product graph.
Unfortunately, the application of the classical PFD approach to this problem 
 is strictly limited, since almost all graphs are prime, although they can have 
a product-like structure. 
In fact, even a very small perturbation, 
such as the deletion or insertion of a single edge, 
  can destroy the product structure completely, modifying a product graph
  to a prime graph \cite{Fei-86,Zmazek:07}. 

The recognition of approximate products has been investigated
by several authors, see e.g. \cite{FEHA-89, HIKS-08, HIKS-09, IZ-96, Zmazek:07,IPZ-97, Zer00, ZZ-02,
HIK-13,hos14-equiparty,HIK-15,HOS-14-relSqP}. 
In \cite{IZ-96} and \cite{Zmazek:07} the authors
showed that  Cartesian and strong product graphs 
can be uniquely reconstructed from each 
of its one-vertex-deleted subgraphs.  
Moreover, in  \cite{IZZ-01} it is shown that
 $k$-vertex-deleted Cartesian product graphs
can be uniquely reconstructed if they have at least $k+1$ factors
and each factor has more than $k$ vertices.
A polynomial-time algorithm for the reconstruction of 
one-vertex-deleted Cartesian product graphs  is given
in \cite{HZ-99}.
In \cite{IPZ-97, Zer00, ZZ-02} algorithms for the recognition of so-called graph bundles are
provided. Graph bundles generalize the notion of graph products and can also be
considered as approximate products.

Another systematic investigation into approximate product graphs
showed that a further practically viable approach can be based on \emph{local} factorization
algorithms, that cover a graph by factorizable small patches 
and attempt to stepwisely extend regions with product structures.
This idea has been fruitful in particular for the strong product of graphs, 
where one benefits from the fact that the local product structure of neighborhoods 
is a refinement of the global factors \cite{HIKS-08, HIKS-09}.
In \cite{HIKS-08} the class of thin-neighborhood intersection coverable
(NICE) graphs was introduced, and a quasi-linear time local 
factorization algorithm
w.r.t. the strong product was devised. 
In \cite{HIKS-09} this approach 
was extended to a larger class of thin graphs which are whose 
local factorization is not finer than the global one, so-called
locally unrefined graphs. 

In this contribution the results of \cite{HIKS-08} and \cite{HIKS-09}  
will be extended and generalized. The main result will
be a new quasi-linear time local prime factorization algorithm
w.r.t. the strong product that works for \emph{all} graph classes.
Moreover, this algorithm can be adapted for the recognition
of approximate products.  This new PFD algorithm 
is implemented in \emph{C++} and  
the source code can be downloaded from 
\url{http://www.bioinf.uni-leipzig.de/Software/GraphProducts}.

This paper is organized as follows.
First, we introduce the necessary basic definitions and give a
short overview of the "classical" prime factor decomposition algorithm
w.r.t. the strong product, that will be slightly modified and 
used locally in our new algorithm. 
The main challenge will be
the combination and the utilization
of the "local factorization information" to derive the global factors. 
To realize this purpose, 
we are then concerned with several important tools and techniques.  
As it turns out, \emph{S-prime} graphs, the so-called \Scond, the \emph{backbone} $\bone(G)$ 
of a graph $G$ and the \emph{\cc}\ property  
will play a central role.
After this, we will derive a new 
general local approach for the prime factor decomposition
for arbitrary graphs, using the previous findings.
Finally, we discuss approximate graph products 
and explain how the new local factorization algorithm
can be modified 
for the  recognition of approximate graph products.

\section{Preliminaries}

\subsection{Basic Notation}
We only consider finite, simple, connected and undirected graphs 
$G=(V,E)$ with vertex set $V$ and edge set $E$. 
A graph is \emph{nontrivial} if it has at least two vertices. 
We define the \emph{$k$-neighborhood} of vertex $v$ as the set 
$N_k[v] = \{x\in V(G)\mid d(v,x) \leq k\}$, where $d(x,v)$
denotes the length of a shortest path connecting the vertices $x$ and $v$.  
Unless there is a risk of confusion, 
we call a 1-neighborhood $N_1[v]$ just neighborhood,
denoted by $N[v]$. 
To avoid ambiguity, we sometimes
write $N^G[v]$ to indicate that $N[v]$ is taken
with respect to $G$.

The degree $\deg(v)$ of a vertex $v$ is the number of 
adjacent vertices, or, equivalently, the number of incident edges. 
The maximum degree in a given graph is denoted by $\Delta$. 
If for two graphs $H$ and $G$ holds 
$V (H) \subseteq V (G)$ and $E(H) \subseteq E(G)$ then $H$ is a 
called a \emph{subgraph} of $G$, denoted by $H\subseteq G$.
If $H\subseteq G$ and all pairs of adjacent vertices in $G$  
are also adjacent in $H$ then $H$ is called an  \emph{induced} subgraph.
The subgraph of a graph $G$ that is induced
by a vertex set $W \subseteq V(G)$ is denoted by $\langle W \rangle$.
A subset $D$ of $V(G)$ is a \emph{dominating set} for $G$, 
if for all vertices in $V(G)\setminus D$ there is at least one
adjacent vertex from $D$. We call $D$ \emph{connected dominating set}, 
if $D$ is a dominating set and the subgraph $\langle D \rangle$ is
connected.

\subsection{Graph Products}

The vertex set of the \emph{strong product} $G_1\boxtimes G_2$  of
two graphs $G_1$ and $G_2$ is defined as 
$V(G_1)\times V(G_2) = \{(v_1,v_2)\mid v_1\in V(G_1), v_2\in V(G_2)\},$
Two vertices $(x_1,x_2)$, $(y_1,y_2)$ are adjacent
in $G_1\boxtimes G_2$ if one of the following conditions is
satisfied:
\begin{itemize} 
 \item[(i)] $(x_1,y_1)\in E(G_1)$ and $x_2=y_2$,\vspace{-0.1in} 
 \item[(ii)]$(x_2,y_2)\in E(G_2)$ and $x_1 = y_1$,\vspace{-0.1in} 
 \item[(iii)] $(x_1,y_1)\in E(G_1)$ and $(x_2,y_2)\in E(G_2)$.
\end{itemize}
The \emph{Cartesian product} $G_1  \Box G_2$ has the same vertex
set as $G_1\boxtimes G_2$, but vertices are only adjacent if they
satisfy (i) or (ii). 
Consequently, the edges of a strong product
that satisfy (i) or (ii) are called \emph{Cartesian}, the others
\emph{non-Cartesian}. The definition of the edge sets shows that the Cartesian product
is closely related to the strong product and indeed it
plays a central role in the factorization 
of the strong products. 

The one-vertex complete graph $K_1$ 
serves as a unit for both products, as $K_1 \Box H = H$ and
$K_1 \boxtimes H = H$ for all graphs $H$.  
It is well-known that both products are associative
and commutative, see \cite{IMKL-00}.
Hence a vertex $x$ of the Cartesian product $\Box_{i=1}^n G_i$, 
respectively the strong product $\boxtimes_{i=1}^n G_i$ 
is properly ``coordinatized'' by the
vector $(x_1,\dots,x_n)$ whose entries are the vertices
$x_i$ of its factor graphs $G_i$.  
Two adjacent vertices in a Cartesian
product graph, respectively endpoints of
a Cartesian edge in a strong product, 
therefore differ in exactly one coordinate.

The mapping $p_j(x) = x_j$ of a vertex $x$ with coordinates 
 $(x_1,\dots,x_n)$ is called \emph{projection} of $x$ onto the $j-th$ factor.
For a set $W$ of vertices of $\Box_{i=1}^n G_i$, resp. $\boxtimes_{i=1}^n G_i$, 
we define $p_j(W) = \{p_j(w)\mid w\in W\}$.
Sometimes we also write $p_A$  if we mean the projection onto factor $A$.

In both products $\Box_{i=1}^n G_i$ and $\boxtimes_{i=1}^n G_i$, 
a \emph{$G_j$-fiber} or \emph{$G_j$-layer} through vertex
$x$ with coordinates  $(x_1,\dots,x_n)$ is the vertex induced subgraph $G_j^x$ in $G$
with vertex set
$\{(x_1,\dots x_{j-1},v,x_{j+1},\dots,x_n) \in V(G)\mid v \in
V(G_j)\}.$ Thus, $G_j^x$ is isomorphic to the factor $G_j$ for every $x\in V(G)$.  
For $y \in V(G_j^x)$ we have $G_j^x = G_j^y$, while $V(G_j^x) \cap V(G_j^z) = \emptyset$
if $z\notin V(G_j^x)$.
Edges of (not necessarily different) $G_i$-fibers are said to be 
edges \emph{of one and the same} factor $G_i$.

Note, the coordinatization of a product is equivalent to a (partial) edge
coloring of $G$ in which edges $e=(x,y)$ share the same color $c(e)=k$ if 
 $x$ and $y$ differ only in the value of a single coordinate $k$,
i.e., if $x_i=y_i$, $i\ne k$ and $x_k\ne y_k$. This colors the
\emph{Cartesian edges} of $G$ (with respect to the \emph{given} product
representation).  
It follows that for each color $k$ the set 
$E_k=\{e\in E(G) \mid c(e)=k\}$ of edges with color $k$ spans $G$. 
The connected components of $\langle E_k\rangle$  
are isomorphic subgraphs of $G$.

A graph $G$ is \emph{prime} with respect
to the Cartesian, respectively the strong product,
if it cannot be written as a Cartesian, respectively 
a strong product, of two nontrivial graphs,
i.e., the identity $G = G_1 \star G_2$ ($\star = \Box,\boxtimes$)
implies that $G_1\simeq K_1$ or $G_2 \simeq K_1$.

As shown by Sabidussi \cite{Sab-59} and independently by Vizing \cite{Viz-63}, 
all finite connected graphs have a unique PFD 
 with respect to the Cartesian product.
The same result holds also for the strong product, as
shown by D{\"o}rfler and Imrich \cite{DOeIM-69} and independently by McKenzie \cite{McK-71}.

\begin{thm} 
	Every connected graph has a unique representation
	as a Cartesian product, resp. a strong product,  
	of prime graphs, up to isomorphisms
	and the order of the factors.
\label{thm:uniquePFD}
\end{thm}

\subsection{Thinness} 

\begin{figure}[tb]
  \centering
  \includegraphics[bb= 212 580 456 686, scale=0.9]{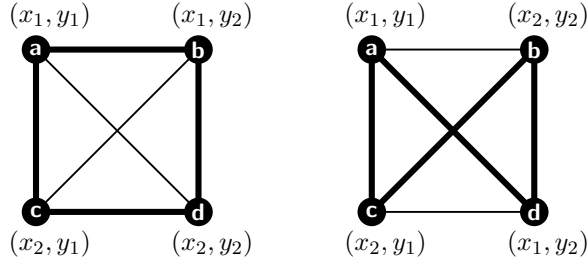}
  \caption[Non-uniquely determined Cartesian edges]{The edge $(a,b)$ is Cartesian in the left, 
													and non-Cartesian in the right coordinatization}
  \label{fig:CartEdges}
\end{figure}

It is important to notice that although the PFD w.r.t. the 
strong product is unique, the coordinatizations might not be.
Therefore, the assignment of an edge being Cartesian or non-Cartesian
is not unique, in general.
Figure \ref{fig:CartEdges} shows that the reason for the
non-unique coordinatizations is the existence of
automorphisms that interchange the vertices $b$ and $d$,
but fix all the others. This is possible because $b$ and $d$
have the same 1-neighborhoods. Thus,  
an important issue in the context of strong graph products 
is whether or not two vertices can be distinguished by their neighborhoods.
This is captured by the relation $S$
defined on the vertex set of $G$, which
was first introduced by D{\"o}rfler and Imrich \cite{DOeIM-69}.
This relation is essential in the studies of the strong product.

\begin{defn}
Let $G$ be a given graph  and  $x, y \in V(G)$ be arbitrary vertices. 
The vertices $x$ and $y$ are in   relation $S$ 
if $N[x] = N[y]$. 
 A graph is $S$-\emph{thin}, or \emph{thin} for short, 
if no two vertices are in relation $S$. 
\end{defn}

In \cite{FESC-92}, vertices $x$ and $y$ with
$xSy$ are called \emph{interchangeable}.
Note that $xSy$ implies that $x$ and $y$ are adjacent since, by definition,
$x\in N[x]$ and $y\in N[y]$. Clearly, $S$ is an equivalence relation. 
The graph $G/S$ is the usual \emph{quotient graph}, more precisely, 
$G/S$ has vertex set $V(G/S) = \{S_i\mid S_i \text{ is an equivalence class of }S \text{ in } G\}$
 and $(S_i,S_j) \in E(G/S)$ whenever $(x,y) \in E(G)$ for some $x\in S_i$ and $y \in S_j$.

Note that the relation $S$ on $G/S$ 
is trivial, that is, its  equivalence classes are single vertices \cite{IMKL-00}.
Thus $G/S$ is thin.
The importance of thinness lies in the 
uniqueness of the coordinatizations, i.e.,
the property of an edge being Cartesian or not
does not depend on the choice of the coordinates.
As a consequence, the Cartesian edges are uniquely 
determined in an $S$-thin graph, see \cite{DOeIM-69, FESC-92}.

\begin{lem} 
	If a graph $G$ is thin, 
	then the set of Cartesian edges
	is uniquely determined and hence the coordinatization is unique.
	\label{lem:uniqu-coord}
\end{lem}

\begin{figure}[tb]
  \centering
  \includegraphics[bb= 263 627 591 718, scale=0.95]{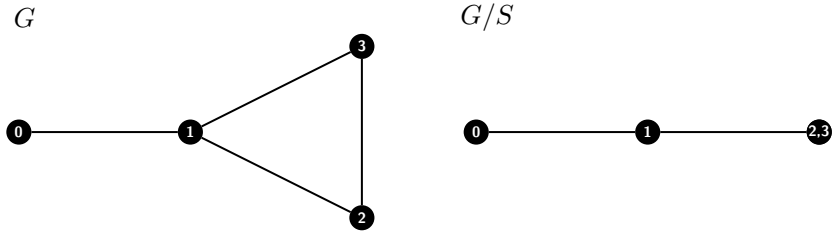}
  \caption[Quotient graph]{A  graph $G$ and its quotient graph $G/S$.
			The S-classes are
                $S_G(0) = \{0\}$, $S_G(1) = \{1\}$, and
                $S_G(2) = S_G(3) = \{2,3\}$. }
  \label{fig:quotient-graph}
\end{figure}

Another important basic property, first proved by D{\"o}rfler and Imrich \cite{DOeIM-69},
concerning the thinness of graphs is stated in the next lemma. 
Alternative proofs can be found in \cite{IMKL-00}.

\begin{lem}
  Let $S_G(v)$ denote the $S$-class in graph $G$ that contains  
  vertex $v$.
  For any two graphs $G_1$ and $G_2$ holds 
  $(G_1 \boxtimes G_2)/S \simeq G_1/S \boxtimes G_2/S$ 
  and for every vertex $x=(x_1,x_2) \in V(G)$ holds
  $S_G(x) = S_{G_1}(x_1) \times S_{G_2}(x_2)$.\\
  Thus, a graph is thin if and only if all of its factors with respect
    to the strong product are thin.
 \label{lem:IMKL-00-quotiengraph}
\end{lem}

\subsection{The Classical PFD Algorithm}
\label{subsec:classicalPFD}

In this subsection,  
we give a short overview of the
classical PFD algorithm that is used locally later on. 

The key idea of finding the PFD
of a graph $G$ with respect to the strong product is 
to find the PFD of a subgraph $\skel(G)$ of $G$, the so-called 
\emph{Cartesian skeleton}, with respect to the Cartesian
product and construct the prime factors of $G$
using the information of the PFD of $\skel(G)$.

\begin{defn}
	A subgraph $H$ of a graph $G = G_1\boxtimes G_2$ with 
	$V(H) = V(G)$ is 
	called \emph{Cartesian skeleton} of $G$,
	if it has a representation  $H = H_1\Box H_2$ such
	that $V(H_i^v) = V(G_i^v)$ 
	for all $v \in V(G)$ and $i\in\{1,2\}$. 
	The Cartesian skeleton $H$ is denoted by $\skel(G)$.
\label{def:CartSk}
\end{defn}

\begin{figure}[tp]
  \centering
  \includegraphics[bb= 200 620 320 745, scale=1.1]{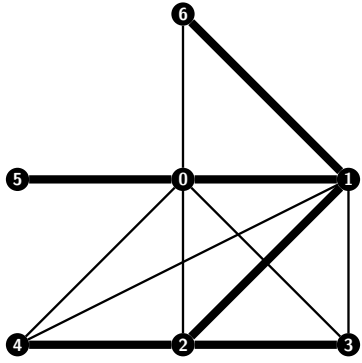}
  \caption[Cartesian skeleton]{
			A prime graph $G$ and its Cartesian Skeleton $\skel(G)$ induced by thick-lined edges.
			Thin-lined edges are marked as dispensable in the approach of Hammack and Imrich.
			On the other hand, the thick-lined edges are marked as Cartesian in the approach
			of Feigenbaum and Sch{\"a}ffer. 
			However, in both cases the resulting Cartesian skeleton $\skel(G)$
			spans $G$. Hence, the vertex sets of the
			$\skel(G)$-fiber (w.r.t. Cartesian product) and the $G$-fiber 
			(w.r.t. strong product) 
			induce the same partition 
			$V(\skel(G)) =V(G)$ of the respective vertex sets.}
  \label{fig:CartSkel}
\end{figure}

In other words, the $H_i$-fibers of the 
Cartesian skeleton $\skel(G)  = H_1 \Box H_2$
of a graph $G = G_1\boxtimes G_2$ induce the
same partition as the $G_i$-fibers on the 
vertex sets $V(\skel(G)) =V(G)$. 
As Lemma \ref{lem:uniqu-coord}
implies, if a graph $G$ is thin then the
set of Cartesian edges and therefore $\skel(G)$ is uniquely determined. 
The remaining question is: How can one determine $\skel(G)$?

The first who answered this question were
Feigenbaum and Sch{\"a}ffer \cite{FESC-92}.
In their polynomial-time approach, edges are marked as Cartesian
if the neighborhoods of their endpoints
fulfill some (strictly) maximal conditions
in collections of neighborhoods or subsets
of neighborhoods in $G$. 

The latest and fastest approach for the detection
of the Cartesian skeleton is
due to Hammack and Imrich  \cite{HAIM-09}. 
In distinction to the approach of Feigenbaum and Sch{\"a}ffer 
edges are marked as dispensable. 
All edges that are dispensable will be removed from $G$. The resulting
graph $\skel(G)$ is the desired Cartesian skeleton 
and will be decomposed with respect to the Cartesian product.
For an example see Figure \ref{fig:CartSkel}.

\begin{defn}
	An edge $(x,y)$ of $G$ is \emph{dispensable} if there exists
	a vertex $z \in V(G)$ for which both of the following statements hold.
	\begin{enumerate}	
		\item (a) $N[x] \cap N[y] \subset N[x] \cap N[z]$ \ \  or \ \  (b) $N[x] \subset N[z] \subset N[y]$
				\vspace{-0.1in} 
		\item (a) $N[x] \cap N[y] \subset N[y] \cap N[z]$ \ \ or \ \ (b) $N[y] \subset N[z] \subset N[x]$
 \end{enumerate}
	\label{def:hammack}
\end{defn}

Some important results, concerning the Cartesian skeleton are summarized 
in the following theorem.

\begin{thm}[\cite{HAIM-09}]
	Let $G = G_1 \boxtimes G_2$ be a strong product graph. 
	If $G$ is connected, then $\skel(G)$ is connected.
	Moreover, if $G_1$ and $G_2$ are thin graphs then $$\skel(G_1\boxtimes G_2) = \skel(G_1) \Box \skel(G_2).$$
	Any isomorphism $\varphi: G \rightarrow H$, as a map $V (G) \rightarrow V (H)$, is also an
	isomorphism $\varphi : \skel(G) \rightarrow \skel(H)$.
	\label{the:CartSk-hamm}
\end{thm}

\begin{rem}
Notice that the set of all Cartesian edges in a strong product 
$G = \boxtimes_{i=1}^n G_i$ of connected, thin prime graphs 
are uniquely determined and hence its Cartesian skeleton is.
 Moreover, since by Theorem \ref{the:CartSk-hamm} and Definition \ref{def:CartSk} 
of the Cartesian skeleton $\skel(G) = \Box_{i=1}^n \skel(G_i)$ 
of $G$ we know that $V(\skel(G)_i^v) = V(G_i^v)$ for all $v \in V(G)$. 
Thus, we can assume without loss of generality
that the set of \emph{all} Cartesian edges in a strong product 
$G = \boxtimes_{i=1}^n G_i$ of connected, thin 
graphs is the edge set of the Cartesian skeleton $\skel(G)$ of $G$. 
As an example consider the graph $G$ in Figure \ref{fig:CartSkel}. 
The edges of the Cartesian skeleton are highlighted by thick-lined edges and 
one can observe that not all edges of $G$ are determined as Cartesian. 
As it turns out $G$ is prime and hence, after the factorization of $\skel(G)$, 
\emph{all} edges of $G$ are determined as Cartesian belonging to a single factor. 
\end{rem}

\begin{figure}[htbp]
  \centering
  \includegraphics[bb= 126 502 475 744, scale=0.9]{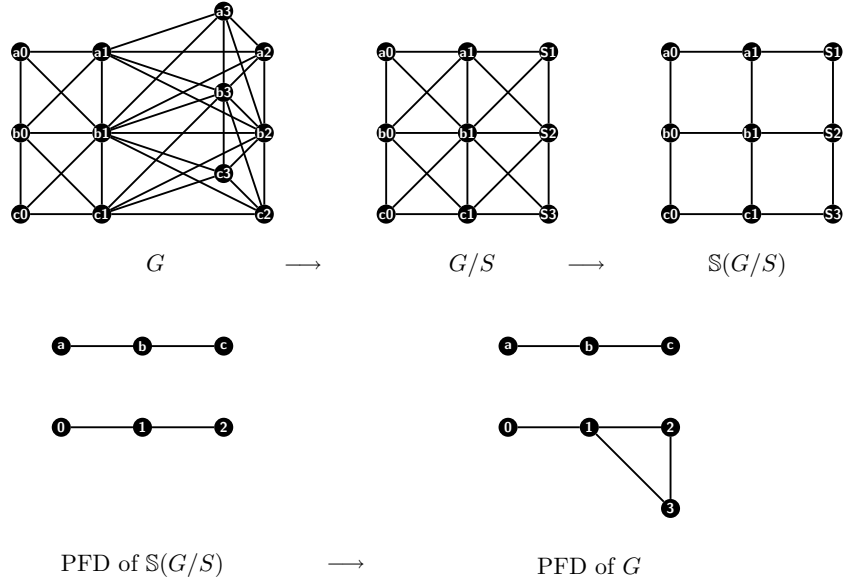}
  \caption[Basic steps of the global approach]{Illustrated are the basic steps of the PFD of strong product graphs.}
  \label{fig:genGlobalAppr}
\end{figure}

Now, we are able to give a brief overview of the global approach that
decomposes given graphs into their prime factors 
with respect to the strong product, see also Figure \ref{fig:genGlobalAppr}.

Given an arbitrary graph $G$, one first 
extracts a possible complete factor $K_l$ of maximal size, resulting in a graph $G'$, i.e.,
$G\simeq G'\boxtimes K_l$, and 
computes the quotient graph $H = G'/S$. This graph $H$ is thin and therefore the
Cartesian edges of $\skel(H)$ can be uniquely determined. 
Now, one computes the prime factors of $\skel(H)$ with respect
to the Cartesian product and utilizes this information to determine 
the prime factors of $G'$ by usage of an additional operation based on 
$gcd$'s of the size of the S-classes, see Lemma 5.40 and 5.41 provided in \cite{IMKL-00}. 
Notice that $G\simeq G'\boxtimes K_l$. 
The prime factors of $G$ are then the prime factors of $G'$ together
with the complete factors $K_{p_1},\dots,K_{p_j}$,
where $p_1\dots p_j$ are the prime factors of the
integer $l$. Figure \ref{fig:genGlobalAppr} gives an overview of 
the classical PFD algorithm.

One can bound the time complexity of this PFD algorithm 
as stated in the next Lemma, see \cite{HAIM-09} and \cite{IKH-11}.

\begin{lem}[\cite{IKH-11}]
	The PFD of a given graph $G$ with $n$ vertices and $m$ edges  
	can be computed	in $O(\max(mn\log n, m^2))$ time.
	\label{lem:complexity_global2a}
\end{lem}

\section{The Local Way to Go - Tools}

As mentioned, we will utilize the classical PFD algorithm
and derive a new approach for the PFD w.r.t. the strong product 
that makes only usage of small subgraphs, so-called \emph{subproducts} of
particular size, and that exploits the local information in order to derive the global
factors. Moreover, motivated by the fact that most graphs are prime, although 
they can have a product-like structure, we want to vary 
this approach such that also disturbed products can be recognized.
The key idea is the following: We try to cover
a given disturbed product $G$ by subproducts that are itself ''undisturbed''. If the graph
$G$ is not too much perturbed, we would expect to be able to cover most of it by 
factorizable $1$-neighborhoods or
other small subproducts and to use this information for the construction of a strong product $H$ that
approximates $G$. 

However, for the realization of this idea several important tools are needed.
First, we give an overview of the subproducts that will be used.
We then introduce the so-called \Scond, that is a property of an edge  
that allows us to determine Cartesian edges, even if the given graph is not thin. 
We continue to examine a subset of the vertex set of
a given graph $G$, the so-called \emph{backbone} $\bone(G)$.
Both concepts, the \Scond\ and the backbone, have first been
investigated in \cite{HIKS-09}. We will see that the backbone is closely related to
the \Scond. 
Finally, in order to identify locally determined fiber as belonging to
one and the same or to different global factors, the so-called 
\emph{color-continuation} property will be introduced. As it
turns out, this particular property is not always met. Therefore,
we continue to show how one can solve this problem for thin 
and later on for non-thin (sub)graphs.

\subsection{Subproducts}

	In this subsection, we are concerned with so-called
	\emph{subproducts}, also known as \emph{boxes} \cite{Ta-97},
	that will be used in the algorithm. 	
	
\begin{defn}
A \emph{subproduct} of a product $G\boxtimes H$, resp. $G\Box H$,
is defined as the strong product, resp. the Cartesian product, 
of subgraphs of $G$ and $H$, respectively. 
\end{defn}

As  shown in \cite{HIKS-08}, it holds that $1$-neighborhoods in strong product graphs are subproducts:

\begin{lem}[\cite{HIKS-08}]
For any two graphs $G$ and $H$ holds
$\langle N^{G\boxtimes H}[(x, y)] \rangle =
   \langle N^G[x]\rangle \boxtimes \langle N^H[y]\rangle$.
\label{lem:N_subproduct}
\end{lem}

\begin{figure}[htbp]
  \centering
  \includegraphics[bb= 192 610 380 728, scale=0.85]{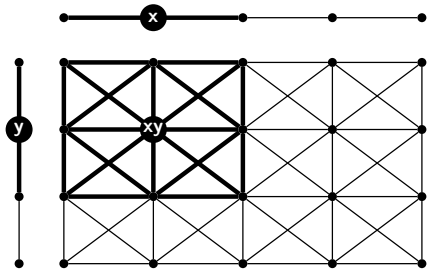}
  \caption[1-neighborhood]{The 1-neighborhood $\la N[(x,y)] \ra = \la N[x] \ra \boxtimes \la N[y] \ra$ is 				
           highlighted by thick lined edges}
\label{fig:subproduct-N}
\end{figure}

	For applications to approximate products it would be desirable
	to use small subproducts. Unfortunately, 
	it turns out that 1-neighborhoods, which would be small enough
	for our purpose, are not sufficient to cover a given graph 
	in general	while providing enough information
	to recognize the global factors. However, 
	we want to avoid to use 2-neighborhoods, although they are 
	subproducts as well, they have diameter $4$ and are thus quite large.
	Therefore, we will define  
	further small subgraphs, that 
	are smaller than 2-neighborhoods, 
	and  show that they	are also subproducts.

\begin{defn}
	Given a graph $G$ and an arbitrary edge $(v,w) \in E(G)$.  The
	\emph{edge-neighborhood} of $(v,w)$ is defined as 
	$$\langle N[v] \cup N[w] \rangle $$ 
    and the \emph{$N^*_{v,w}$-neighborhood} is defined as 
	$$N^*_{v,w} = \langle \bigcup_{x\in N[v] \cap N[w]} N[x]\rangle.$$
	\label{def:N}
\end{defn}

	If there is no risk of confusion we will denote $N^*_{v,w}$-neighborhoods
	just by $N^*$-neighborhoods.
	We will show in the following that in addition to 1-neighborhoods
	also edge-neighborhoods of Cartesian edges and $N^*$-neighborhoods 
	are subproducts	and hence, natural candidates to cover a given graph as well. 
	We show first, given a subproduct $H$ of $G$, that the 
	subgraph which is induced by vertices contained 
	in the union of 1-neighborhoods $N[v]$ with $v\in V(H)$, 
	is itself a subproduct of $G$.

\begin{figure}[tbp]
  \centering
  \includegraphics[bb= 192 592 550 710, scale=0.85]{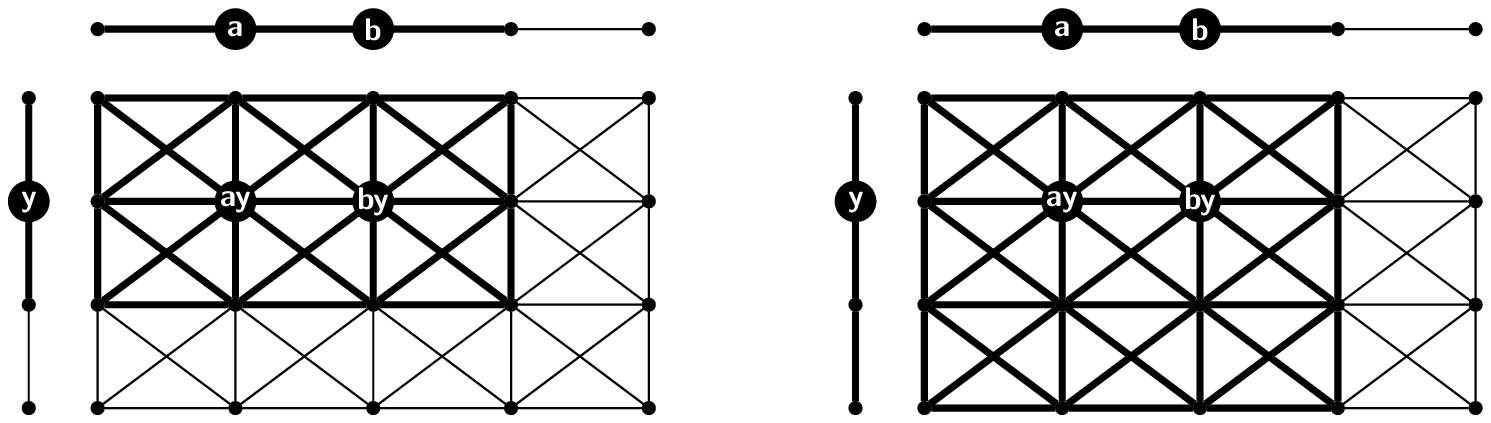}
  \caption[Edge-neighborhood and $N^*$-neighborhood] {Shown is a strong product graph of two paths. Notice that 
			 the 2-neighborhood $\la N_2[(b,y)] \ra$ of vertex $(b,y)$ is isomorphic to $G$.\\
			\textbf{lhs.:} The edge-neighborhood 
			$\langle N[(a,y)] \cup N[(b,y)] \rangle  = \la(N[a]\cup N[b])\ra \boxtimes \la N[y] \ra$.\\
			\textbf{rhs.:}  The $N^*$-neighborhood $N^*_{(a,y),(b,y)}  = \la\cup_{z\in N[a]\cap N[b]} N[z]\ra 							\boxtimes \la \cup_{z \in N[y]} N[z]\ra$.}
\label{fig:subproduct-edgeN-N}
\end{figure}

\begin{lem}
	Let $G= G_1 \boxtimes G_2$ be a strong product graph and 
	$H= H_1 \boxtimes H_2$ be a subproduct of $G$.
	Then $$H^* = \left\langle \cup_{v\in V(H)} N^G[v] \right\rangle$$ 
	is a subproduct of $G$ with  $H^* = H_1^* \boxtimes H_2^*$, where 
	$H_i^*$ is the induced subgraph of factor $G_i$ on the vertex set
	$V(H_i^*)= \bigcup_{v_i\in V(H_i)} N^{G_i}[v_i]$, $i=1,2$.
	\label{lem:H_union_N_is_product}
\end{lem}

\begin{proof}
	It suffices to show that $V(H^*) = V(H_1^*)\times V(H_2^*)$.
	For the sake of convenience, we denote $V(H_i)$ by $V_i$, for $i=1,2$.
 	We have:
		$  V(H^*) = \bigcup_{v\in V(H)} N^G[v] 	  = \bigcup_{v\in V_1\times V_2} N^G[v].$

		Since the induced neighborhood of each vertex $v = (v_1,v_2)$ 
		in $G$ is the product of the corresponding neighborhoods $ N^{G_1}[v_1] \boxtimes N^{G_2}[v_2] $
		we can conclude:

			$ V(H^*) = \bigcup_{\{v_1\in V_1\}\times \{v_2 \in V_2\}} ( N^{G_1}[v_1]\times N^{G_2}[v_2] )
				   = \bigcup_{v_1\in V_1}  N^{G_1}[v_1] \times \bigcup_{v_2\in V_2} N^{G_2}[v_2] 
				   = V(H_1^*) \times V(H_2^*) $
\end{proof}

\begin{lem}
	Let $G$ be a nontrivial strong product graph and 
	$(v,w)$ be an arbitrary edge of $G$. 
	Then $\langle N^G[v] \cap N^G[w] \rangle$ is a subproduct.
	\label{lem:intersection_neigb_is_product}
\end{lem}
\begin{proof}
	Let $v$ and $w$ have coordinates $(v_1,v_2)$ and $(w_1,w_2)$, respectively.
	Since $ N^G[v] = N^{G_1}[v_1] \times N^{G_2}[v_2]$ we can conclude that 
		$N^G[v] \cap N^G[w] = (N^{G_1}[v_1] \times N^{G_2}[v_2]) \cap (N^{G_1}[w_1] \times N^{G_2}[w_2])   
					   = (N^{G_1}[v_1] \cap N^{G_1}[w_1]) \times (N^{G_2}[v_2] \cap N^{G_2}[w_2]). $
\end{proof}

Lemmas \ref{lem:N_subproduct}, \ref{lem:H_union_N_is_product} and \ref{lem:intersection_neigb_is_product} 
directly imply the next corollary.

\begin{cor}
	Let $G$ be a given graph. Then for all $v \in V(G)$ and all edges $(v,w) \in E(G)$
	holds:
	 $$\langle N_2[v] \rangle \text{ and } N^*_{v,w}$$ are subproducts of $G$.
    Moreover, if the edge $(v,w)$ is Cartesian than the 
	edge-neighborhood $$\langle N[v] \cup N[w] \rangle$$ is a subproduct of $G$.
	\label{cor:boxes}
\end{cor}

Notice that $\langle N[v] \cup N[w] \rangle$ could be a product, i.e., not prime, 
even if $(v,w)$ is non-Cartesian in $G$. 
However, the edge-neighborhood of a single non-Cartesian edge is not 
a subproduct, in general. 
The obstacle we have is that a non-Cartesian edge of $G$ might be 
Cartesian in its edge-neighborhood.
Therefore, we cannot use the information provided by the PFD of
$\la N[x] \cup N[y] \ra$ to figure out if $(x,y)$ is Cartesian in $G$ and 
hence, if $\la N[x] \cup N[y] \ra$ is a proper subproduct.
On the other hand, an edge that is Cartesian in a subproduct $H$
of $G$ must be Cartesian in $G$.
To check if an edge $(x,y)$ is Cartesian
in $\la N[x] \cup N[y] \ra$ that is Cartesian in $G$ as well
we use the \emph{dispensable}-property provided
by Hammack and Imrich, see \cite{HAIM-09}.

We show that an edge $(x,y)$  that is dispensable in $G$ 
is also dispensable in $\la N[x] \cup N[y] \ra$. Conversely, we can
conclude that every edge that is indispensable in
 $\la N[x] \cup N[y] \ra$ must be indispensable and
therefore Cartesian in $G$. This implies that
every edge-neighborhood  $\la N[x] \cup N[y] \ra$
is a proper subproduct of $G$ if $(x,y)$ is indispensable
in $\la N[x] \cup N[y] \ra$.

\begin{rem} As mentioned in \cite{HAIM-09}, we have:
	\begin{itemize} 
   		\item	$N[x] \subset N[z] \subset N[y]$ implies $N[x] \cap N[y] \subset N[y] \cap N[z]$.\vspace{-0.1in} 
		\item	$N[y] \subset N[z] \subset N[x]$ implies $N[x] \cap N[y] \subset N[x] \cap N[z]$.\vspace{-0.1in} 
		\item	If $(x,y)$ is indispensable then $N[x] \cap N[y] \subset N[x] \cap N[z]$ and
				$N[x] \cap N[y] \subset N[y] \cap N[z]$ cannot both be true.
	\end{itemize}
	\label{rem:disp_and_conclusion}
\end{rem}

By simple set theoretical arguments one can easily 
prove the following lemma. 

\begin{lem}
	Let $(x,y)$ be an arbitrary edge of  a given graph $G$ and $H = \langle N[x]\cup N[y] \rangle$.  
	Then it holds:	
	$$N[x] \cap N[y] \subset N[x] \cap N[z] \Leftrightarrow
	N[x] \cap N[y] \cap H \subset N[x] \cap N[z] \cap H$$ 
	and 
	$$N[x] \subset N[z] \subset N[y] \Rightarrow 
	N[x] \cap H  \subset N[z] \cap H \subset N[y] \cap H$$
	\label{lem:hammack_global_local_edgeN} 
\end{lem}

Notice that the converse of the second statement does not hold in general, since 
$N[z] \cap H \subset N[y] \cap H  = N[y]$ does not imply 
that $N[z] \subset N[y]$.
However, by symmetry, Remark \ref{rem:disp_and_conclusion}, Corollary \ref{cor:boxes}, Lemma 
\ref{lem:hammack_global_local_edgeN}   
we can conclude the next 
corollary.

\begin{cor}
If an edge $(x,y)$ of a thin strong product graph $G$ 
is \emph{indispensable} in $\langle N[x]\cup N[y] \rangle$ and therefore
Cartesian in $G$ 
then the edge-neighborhood $\langle N[x]\cup N[y] \rangle$
is a subproduct of $G$.
\label{cor:proper_edgeN}
\end{cor}

\subsection{The S1-condition and the Backbone}

The concepts of the \Scond\ and the \emph{backbone} 
were first introduced in \cite{HIKS-09}.
The main idea of our approach is to construct the Cartesian skeleton of $G$ 
by considering PFDs of the introduced subproducts only.
The main obstacle is that
even though $G$ is thin, this is
not necessarily true for subgraphs, Fig. \ref{fig:induced_N_not_thin}. 
Hence, although the Cartesian edges are uniquely determined in $G$,
they need not to be unique in those subgraphs.
In order to investigate this issue in
some more detail, we also define $S$-classes w.r.t.\ subgraphs $H$ of a given graph $G$.

\begin{defn} 
  Let $H\subseteq G$ be an arbitrary induced subgraph of a given graph $G$. 
  Then $S_H(x)$ is defined as the set 
  \begin{equation*}
     S_H(x) = \left\{ v \in V(H) \mid   N^G[v]\cap V(H) = N^G[x]\cap V(H) \right\}.
   \end{equation*}
	If $H = \langle N^G[y] \rangle$ for some $y\in V(G)$ we set 
		$S_y(x) := S_{\langle N^G[y] \rangle}(x)$
   \label{def:S-class}
\end{defn}

In other words, $S_H(x)$ is the $S$-class that contains $x$ in the subgraph
$H$. Notice that $N[x] \subseteq N[v]$ holds for all $v
\in S_x(x)$.  If $G$ is additionally thin, then $N[x] \subsetneq N[v]$.

\begin{figure}[htbp]
  \centering
   \includegraphics[bb= 200 650 440 720]{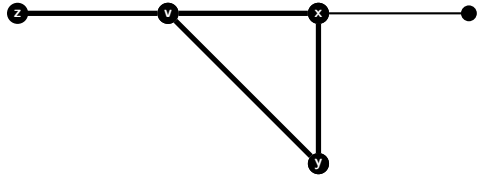}
  \caption[A thin graph with non-thin neighborhood]
			{A thin graph where  $\langle N[v] \rangle$ is not thin.
                The S-classes in $\langle N[v] \rangle$ are
                $S_v(v) = \{v\}$, $S_v(z) = \{z\}$ and
                $S_v(x) = S_v(y) = \{x,y\}$.}
  \label{fig:induced_N_not_thin}
\end{figure}

Since the Cartesian edges are globally uniquely defined in a thin graph,
the challenge is to find a way to determine enough Cartesian edges from
local information, even if $\langle N[v] \rangle$ is not thin. 
This will be captured by the \Scond\ and the \emph{backbone} of graphs.

\begin{defn}
  Given a graph $G$. An edge $(x,y) \in E(G)$
  satisfies the \Scond\ in an induced subgraph $H\subseteq G$ if 
  \begin{enumerate}[(i)]
	  \item $x,y \in V(H)$ and
	  \item $|S_H(x)| = 1$ or $|S_H(y)| = 1$.
  \end{enumerate}
\end{defn}

Note that $|S_H(x)| = 1$ for all $x\in V(H)$, if $H$ is thin.
From Lemma \ref{lem:IMKL-00-quotiengraph} we can directly infer that the
cardinality of an $S$-class in a product graph $G$ is the product
of the cardinalities of the corresponding $S$-classes in the
factors. Applying this fact to subproducts of $G$ 
immediately implies 
Corollary \ref{cor:common_cardi_of_S}.

\begin{cor} 
  Consider a strong product $G = \boxtimes_{i=1}^n G_i$ and a subproduct
  $H = \boxtimes_{i=1}^n H_i \subseteq G$. 
   Let $x \in V(H)$ be a given vertex with coordinates
  $(x_1,\dots,x_n)$. Then $S_H(x) =  \times_{i=1}^n S_{H_i}(x_i)$
	and therefore, $|S_H(x)| =  \prod_{i=1}^n |S_{H_i}(x_i)|$.
 \label{cor:common_cardi_of_S}
\end{cor}

The most important property of Cartesian edges that satisfy the \Scond\ in
some quotient graph $G/S$ is that they can be identified as 
Cartesian edges in $G$, even if $G$ is not thin. 

\begin{lem}[\cite{HIKS-09}]
  Let $G = \boxtimes_{i=1}^n G_i$ be a strong product graph  containing two S-classes
  $S_G(x)$, $S_G(y)$ that satisfy
  \begin{enumerate}[(i)]
  \item $(S_G(x),S_G(y))$ is a Cartesian edge in $G/S$ and 
  \item $|S_G(x)| = 1$ or $|S_G(y)| = 1$.
  \end{enumerate}
  Then all edges in G induced by vertices  of $S_G(x)$ and $S_G(y)$ are
  Cartesian and copies of one and the same factor.
  \label{lem:S=1condition}
\end{lem}

\begin{figure}[tbp]
  \centering 
  \includegraphics[bb= 170 638 459 737, scale =1.3]{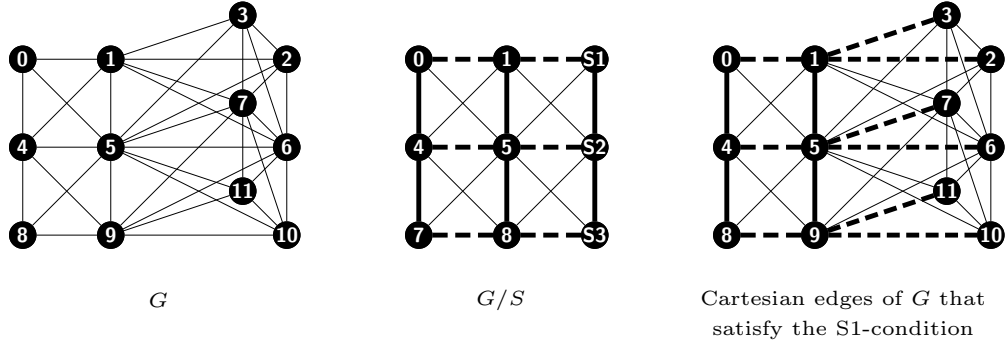}
  \caption[Determining Cartesian edges that satisfy the \Scond]{Determining Cartesian edges that satisfy the \Scond. 				Given a graph $G$, one 
           computes its quotient graph $G/S$. Since $G/S$ is thin
           the Cartesian edges of $G/S$ are uniquely
           determined. Now one factorizes $G/S$ and computes
			the prime factors of $G$. 
			Apply Lemma \ref{lem:S=1condition} to identify all Cartesian edges with respective colors
			(thick and dashed lined)  in $G$ that satisfy the \Scond. The backbone $\bone(G)$ is
			the singleton $\{5\}$.}
  \label{fig:S=1-example}
\end{figure}

\begin{figure}[tbp]
  \centering  
  \includegraphics[bb= 183 638 453 724, scale =1.3]{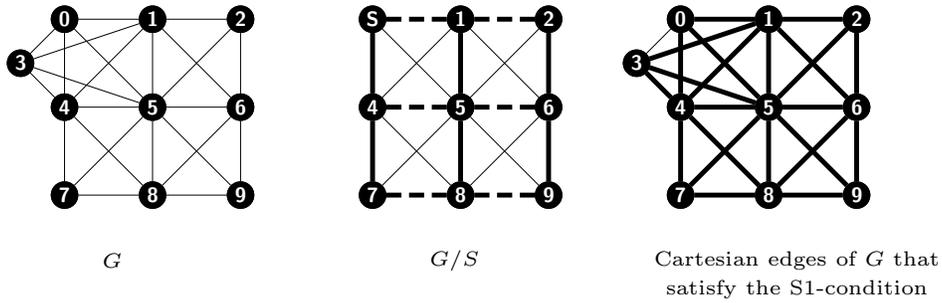}
  \caption[Determining Cartesian edges that satisfy the \Scond]{
			Determining Cartesian edges that satisfy the \Scond.
			We factorize $G/S$ and compute
			the prime factors of $G$. 
			Notice that it turns out that the factors induced
			by thick and dashed lined edges have to be merged to
			one factor. Apply now Lemma \ref{lem:S=1condition}
           to identify all Cartesian edges in $G$ that satisfy the \Scond. 
		In this case, it is clear that the edge $(0,3)$ has
			to be Cartesian as well and belongs to the single prime factor
			$G$. The backbone $\bone(G)$ is
			the singleton $\{5\}$.}
  \label{fig:S=1-example2}
\end{figure}

\begin{rem}
Whenever we find a Cartesian edge $(x,y)$ in a subproduct $H$ of $G$ 
such that one endpoint of $(x,y)$ is contained in a $S$-class of
cardinality $1$ in $H / S$, i.e., such that
$S_H(x)=\{x\}$ or $S_H(y)=\{y\}$, we can therefore conclude that all edges
in $H$ induced by vertices of $S_H(x)$ and $S_H(y)$ are
also Cartesian and are copies of one and the same factor, see 
Figure \ref{fig:S=1-example}. 

Note, even if $H / S$ has more factors than
$H $ the PFD algorithm provided by Imrich and Hammack
indicates which factors have to be merged to one factor. 
Again we can conclude that all edges 
in $H$ that satisfy the \Scond\ are
Cartesian and are copies of one and the same factor, 
see Figure \ref{fig:S=1-example2}.  

Moreover, since  $H$ is a subproduct of $G$,  
it follows that any Cartesian edge of 
 $H$ that satisfies the \Scond\ is a Cartesian edge in $G$.
    \label{rem:local_S=1condition_cardinality}
\end{rem}

We consider now a subset of $V(G)$, the so-called \emph{backbone}, 
which is essential for the algorithm.

\begin{defn}
  The \emph{backbone} of a thin graph $G$ is the vertex set 
  $$\bone(G) = \{ v \in V(G)\mid |S_v(v)|=1 \}\,.$$
  Elements of $\bone(G)$ are called \emph{backbone vertices}.
  \label{def:backbone}
\end{defn}

Clearly, the backbone $\bone(G)$ and the \Scond\ are closely related, since 
all  edges $(x,y)$ that contain a backbone vertex, say $x$, 
satisfy the \Scond\ in $\la N[x]\ra$.
If the  backbone $\bone(G)$ of a given graph $G$ is nonempty then 
Corollary \ref{cor:common_cardi_of_S} implies that
 no factor of $G$ is isomorphic to a complete graph, otherwise 
we would have $|S_v(v)|>1$ for all $v\in V(G)$. 
The last observations lead directly to the next corollary.

\begin{cor}
Given a graph $G$ with nonempty backbone $\bone(G)$ then
for all  $v\in \bone(G)$ holds:
all edges $(v,x)\in E(\la N[v]\ra)$ satisfy the \Scond\ in $N[v]$. 
\label{cor:s1}
\end{cor}

The set of backbone vertices of thin graphs can be characterized
as follows.

\begin{lem}[\cite{HIKS-09}]
  Let $G$ be a thin graph and $v$ an arbitrary vertex of $G$.
  Then $v\in \bone(G)$ if and only if $N[v]$ is a strictly maximal
  neighborhood in $G$.
  \label{lem:backbone_strictMax_vertices}
\end{lem}

As shown in \cite{HIKS-09} the backbone B(G) of thin graphs G is a connected dominating set.
This allows us to cover the entire graph by $1$-neighborhoods of the backbone vertices only. 
Moreover, it was shown 
that it suffices to exclusively use information about the $1$-neighborhood of backbone vertices, 
to find all Cartesian edges that satisfy the \Scond\ in arbitrary $1$-neighborhoods, 
even those edges $(x,y)$ with $x, y \notin \bone(G)$. These results are summarized
in the next theorem. 

\begin{thm}[\cite{HIKS-09}]
  Let $G$ be a thin graph. Then the backbone $\bone(G)$ is
  a connected dominating set for $G$.
	\newline
  All Cartesian edges that satisfy the \Scond\ in an arbitrary induced
  1-neighborhood also satisfy the \Scond\ in the induced 1-neighborhood of a
  vertex of the backbone $\bone(G)$.
    \label{thm:backbone_coverable}
  \label{thm:all_GIX-edges-S=1}
\end{thm}

Consider now the subproducts $\la N[x] \ra$,  $ N^*_{x,y}$
and $\la N[x] \cup N[y] \ra$ of a thin graph $G$.
We will show in the following that the set of Cartesian edges of these subproducts 
that satisfy the \Scond, induce a connected subgraph in the respective subproducts. 
This holds even if $\la N[x] \ra$,  $ N^*_{x,y}$
and $\la N[x] \cup N[y] \ra$ are not thin. For this we need the next lemmas.

\begin{lem}
Let G be a given thin graph, $x\in \bone(G)$ and $H\subseteq G$ be an arbitrary induced subgraph 
such that $N[x]\subseteq V(H)$.
Then $|S_H(x)| = 1$ and $x\in \bone(H)$.
\label{lem:s1_in_H}
\end{lem}

\begin{proof}
	First notice that Lemma \ref{lem:backbone_strictMax_vertices}
	and $x\in \bone(G)$ implies that $\langle N[x] \rangle$ 
	is strictly maximal in $G$. Since $\langle N[x] \rangle \subseteq H \subseteq G$ 
	we can conclude that $\langle N[x] \rangle$ is strictly maximal
	in $H$. Hence, it holds $|S_H(x)|=1$. Moreover, it holds $|S_x(x)|=1$, 
	otherwise there would be a vertex $y\in S_x(x),\ y\neq x$ and therefore, 
	$N[x]\subseteq N[y]$. This contradicts
    that $\langle N[x] \rangle$ is strictly maximal in $H$.
	Hence,  $x\in \bone(H)$. 
\end{proof}

\begin{lem}
Let $H=\boxtimes_{i=1}^n H_i$ be an arbitrary connected (not necessarily thin) graph and 
$(x,y)\in E(H)$ such that $|S_H(x)|=|S_H(y)|=1$. 
Then there is a path $\mathcal{P}_{x,y}$ from $x$ to $y$ consisting of Cartesian
edges $(u,w)$ only with 
$|S_H(u)|=|S_H(w)|=1$. 
\label{lem:CartPath-via-S1}
\end{lem}

\begin{proof}
Let $(x,y)$ be an arbitrary edge of $H$ with $|S_H(x)|=|S_H(y)|=1$. 
From Corollary \ref{cor:common_cardi_of_S} we can conclude
that $|S_{H_i}(x_i)| = 1$ and $|S_{H_i}(y_i)| = 1$ for all $i=1,\dots,n$.
If $(x,y)$ is Cartesian there is nothing to show.
Thus, assume $(x,y)$ is a non-Cartesian edge. Hence, 
the coordinates of $x = (x_1, \dots, x_n)$ and 
$y=(y_1, \dots y_n)$ differ in more than
one position. 
W.l.o.g we assume that $x$ and $y$ differ in
the first positions $1,\dots,k$. 
Hence $(x_i,y_i)\in E(G_i)$ for all $i=1,\dots,k$ and
$x_i=y_i$ for all $i=k+1,\dots,n$.
Therefore, one can construct a path $\mathcal{P}_{x,y}$ with edge set 
$\{(y,v^1),(v^1,v^2),\dots,(v^{k-1},x)\}$ 
such that the vertices $v^j$ have respective
coordinates $(x_1,x_2,\dots,x_j,y_{j+1},\dots,y_n)$, 
$j=1,\dots,k-1$. Since all edges have endpoints
differing in exactly one coordinate, all  
edges in $\mathcal{P}_{x,y}$ are Cartesian.  
Corollary \ref{cor:common_cardi_of_S} implies that 
for all those vertices hold $|S_H(v^j)|=1$ and 
hence in particular for all edges 
$(u,w) \in \mathcal{P}_{x,y}$
hold $|S_{H}(u)|=1$ and $|S_{H}(w)|=1$. 
\end{proof}

\begin{lem}[\cite{HIKS-09}]
    Let $G$ be a thin, connected simple graph
    and $v \in V(G)$ with $|S_v(v)| > 1$.
    Then there exists a vertex $y \in S_v(v)$
    s.t. $|S_y(y)| = 1$.
    \label{lem:S=1condition2}
\end{lem}

\begin{lem}
Let G be a given thin graph, $x,y\in \bone(G)$ and let $H\subseteq G$ 
denote  one of the subproducts $\la N[x] \ra$,  $ N^*_{x,y} $
or $\la N[x] \cup N[y] \ra$. In the latter case we assume
that the edge $(x,y)$ is Cartesian in $H$.
Then the set of all Cartesian edges of $H$ that satisfy the \Scond\ in $H$
induce a connected subgraph of $H$.
\label{cor:connectCartSk-via-S1}
\label{lem:connectCartSk-via-S1}
\end{lem}

\begin{proof}
 First, let $H=\la N[x] \ra$. Clearly, it holds  $|S_H(x)| = 1$. 
 Let $(a,b)$ be an arbitrary edge that satisfy the \Scond\ in $H$.
 W.l.o.g. we assume that $|S_H(a)| = 1$. If $(a,x)$ is Cartesian there is 
 nothing to show and if $(a,x)$ is non-Cartesian one can construct
 a path $\mathcal{P}_{x,a}$  as shown in Lemma \ref{lem:CartPath-via-S1}. 

 Second, let $H=\la N[x] \cup N[y] \ra$. Lemma \ref{lem:s1_in_H} implies that
 $|S_H(x)|=|S_H(y)|=1$. Let $(a,b)$ be an arbitrary edge that satisfy the \Scond\ in $H$.
  W.l.o.g. we assume that $|S_H(a)| = 1$. Moreover, let $a\in N[x]$. 
  If  $(a,x)$ is Cartesian there is nothing to show and if $(a,x)$ 
 is non-Cartesian one can construct a path $\mathcal{P}_{x,a}$  as shown in Lemma \ref{lem:CartPath-via-S1}.
 Analogously, one shows that such paths $\mathcal{P}_{y,a}$ can be constructed if  $a\in N[y]$.
 Therefore, all Cartesian edges are connected to $x$ or $y$ via paths consisting
 of Cartesian edges only that satisfy the \Scond. 
 Furthermore $(x,y)$ is Cartesian and thus, the assertion follows for $H=\la N[x] \cup N[y] \ra$. 

 Third, let $H= N^*_{x,y}$. Lemma \ref{lem:s1_in_H} implies that
 $|S_H(x)|=|S_H(y)|=1$. Therefore, one can construct a path $\mathcal{P}_{x,y}$
 as shown in Lemma \ref{lem:CartPath-via-S1}, since $(x,y)\in E(G)$. 
 Let $(a,b)$ be an arbitrary edge that satisfy the \Scond\ in $H$.
 W.l.o.g. we assume that $|S_H(a)| = 1$. If $a\in N[x]$ or $a\in N[y]$ one can show
 by similar arguments as in the latter case that there is a path $\mathcal{P}_{x,a}$, resp.,
 $\mathcal{P}_{y,a}$ consisting of Cartesian edges only that satisfy the \Scond. 
 Assume $a\notin N[x]$ and $a\notin N[y]$. Then there is a vertex $v \in N[x]\cap N[y]$
 such that $a\in N[v]$.
 If $v\in \bone(G)$ then Lemma \ref{lem:s1_in_H} implies 
 that $|S_H(v)| = 1$, since $N[v]\subseteq V(H)$ and one 
 construct a path $\mathcal{P}_{a,v}$ and $\mathcal{P}_{v,x}$ as in Lemma \ref{lem:CartPath-via-S1}.
 Now assume $v\notin \bone(G)$. Theorem \ref{thm:all_GIX-edges-S=1} implies that
 there is a vertex $z\in \bone(G)$ such that $z\in N[v]$. Moreover, as stated in 
 Lemma \ref{lem:S=1condition2}, there exists even a vertex $z\in \bone(G)$ such that $z\in S_v(v)$ and therefore
 $N[v]\cap N[z]=N[v]$. Thus it holds that $a,x,y\in N[z]$ and hence, $N[z]\subseteq H$. Therefore, 
 Lemma \ref{lem:s1_in_H} implies that $|S_H(z)|=1$. Analogously as in Lemma \ref{lem:CartPath-via-S1}, 
one can construct a path $\mathcal{P}_{a,z}$ and $\mathcal{P}_{z,x}$, as well as a path $\mathcal{P}_{z,y}$
consisting of Cartesian edges only that satisfy the \Scond. 
\end{proof}

Last, we state two lemmas for later usage. Note,  
the second lemma refines  the already known results of  \cite{HIKS-09}, where 
analogous results were stated for $2$-neighborhoods.

\begin{lem}[\cite{HIKS-09}]
  Let $(x,y)\in E(G)$ be an arbitrary edge in a thin graph $G$
  such that $|S_x(x)| > 1$.
  Then there exists a vertex $z\in \bone(G)$
  s.t.\ $z\in N[x]\cap N[y]$.
  \label{lem:edges_dont_sat_Scond}
\end{lem}

\begin{figure}[t]
  \centering
  \includegraphics[bb= 132 505 540 791, scale=0.7]{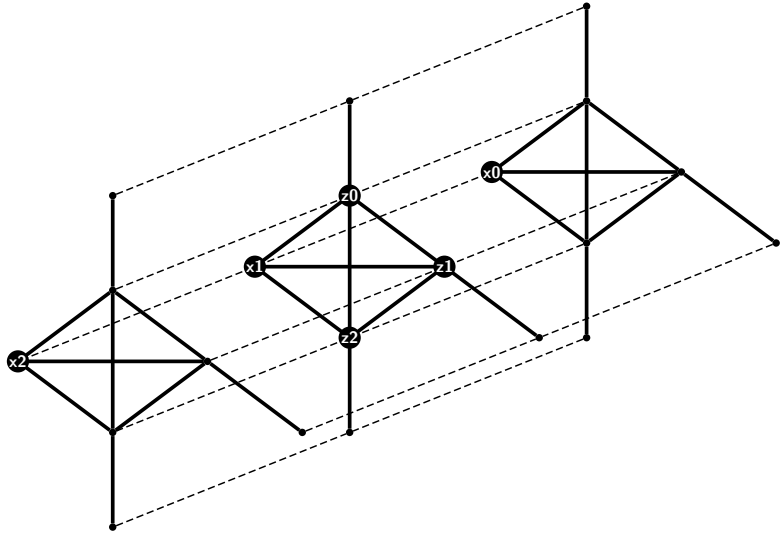}
  \caption[A strong product graph where $N^*$-neighborhoods are necessary]{
			The Cartesian skeleton of the thin product graph $G$ of 
			 two prime factors induced by one connected component of thick and dashed lined edges.
			 The backbone $\bone(G)$ consists of the vertices $z_1,z_2$ and $z_3$.
			 In \emph{none} of a edge-neighborhood $H$ holds $|S_H(x_i)| =1$, $i=1,2,3$.
			 Hence the fiber induced by vertices $x_1,x_2$ and $x_3$ does not satisfy 
			 the \Scond\ in any edge-neighborhood. To identify this particular fiber it is 
			 necessary to use $N^*$-neighborhoods. By Lemma \ref{lem:N*} $N^*$-neighborhoods
			 are also sufficient.}
  \label{fig:Counter_edgeN_gen}
\end{figure}

\begin{lem}
  Let $G$ be a thin graph and $(v,w)$ be any edge of $G$. 
  Let $N^*$ denote the $N^*_{v,w}$-neighborhood. 
	Then it holds that $|S_{N^*}(v)|=1$ and $|S_{N^*}(w)|=1$ , i.e., 
	the edge $(v,w)$ satisfies the \Scond\ in $N^*$.
  \label{lem:N*}
\label{lem:diam4subgraph}
\end{lem}
\begin{proof}
	Assume that $|S_{N^*}(v)|>1$. Thus there 
	is a vertex $x \in S_{N^*}(v)$ different from $v$ 
 	with $N[x] \cap N^* = N[v] \cap N^*$, which implies
    that $w \in N[x]$ and hence, $x\in N[v]\cap N[w]$. 
	Thus, it holds $N[x] \subseteq N^*$.
	Moreover, since  $N[v] \subseteq N^*$ we can conclude that
	$N[v] = N[v] \cap N^* = N[x] \cap N^* = N[x]$,
	contradicting that $G$ is thin. Analogously, one
	shows that the statement holds for vertex $w$.
\end{proof}

\subsection{The Color-Continuation}

The concept of covering a graph by suitable subproducts 
and determining the global factors needs some additional 
improvements. Since we want to determine the global
factors, we need to find their fibers.
This implies that we have to identify different locally determined
fibers as belonging to different or to one and the same global fiber.
For this purpose, we formalize the term \emph{product coloring},  
\emph{color-continuation} and \emph{combined coloring}.
Remind, the coordinatization of a product is equivalent to a (partial) edge
coloring of $G$ in which edges $e=(x,y)$ share the same color $c(e)=k$ if 
 $x$ and $y$ differ only in the value of a single coordinate $k$,
i.e., if $x_i=y_i$, $i\ne k$ and $x_k\ne y_k$. This colors the
\emph{Cartesian edges} of $G$ (with respect to the \emph{given} product
representation).  

\begin{defn}
	A \emph{\PC} of a strong product graph $G = \boxtimes_{i=1}^n G_i$ of $n\geq 1$ 
	(not necessarily prime) factors 
	is a mapping  $P_G$ from a subset $E'\subseteq E(G)$, that is a 
	set of Cartesian edges of $G$, 
	into a set $C =\{1,\dots,n\}$ of colors, such that all such edges in $G_i$-fibers  
	obtain the same color $i$.
 \label{def:product_coloring}
\end{defn}

\begin{defn}
	A \emph{partial \PC} of a graph $G = \boxtimes_{i=1}^n G_i$
	is a \PC\  that is only defined on edges that additionally satisfy the \Scond\ in 
	 $G$.
 \label{def:partial_product_coloring}
\end{defn}

Note, in a thin graph $G$ a \PC\ and a partial \PC\ coincide, since
all edges satisfy the \Scond\ in $G$. 

\begin{defn} \label{continuation}
	Let $H_1,H_2 \subseteq G$ and $P_{H_1}$, resp. $P_{H_2}$, be partial \PC s 
	of $H_1$, resp. $H_2$. Then $P_{H_2}$ is  a 
	\emph{color-continuation} of $P_{H_1}$ if for every color $c$ in the image of
	$P_{H_2}$ there is an edge in $H_2$ with color $c$ that is also in
	the domain of $P_{H_1}$.

	The \emph{combined coloring} on $H_1 \cup H_2$ uses the colors of
	$P_{H_1}$ on $H_1$ and those of $P_{H_2}$ on $H_2 \setminus H_1$.
\end{defn}

\begin{figure}[ht]
  \centering
  \includegraphics[bb= 20 519 589 656, scale=0.6]{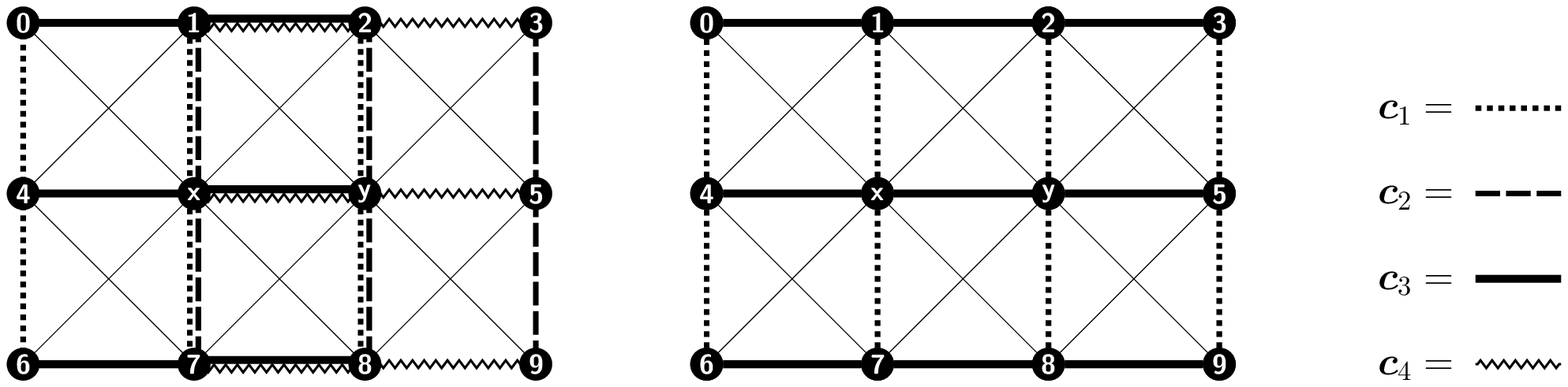}
  \caption[Example: color-continuation works]{Shown is a thin graph $G$ with $\bone(G)=\{x,y\}$.
		$G$ is the strong product of	two paths. 
	If one computes the PFD of the neighborhood $\la N[x]\ra$ one obtains 
		 a (partial) \PC\ with colors $c_1$ and $c_3$. The (partial) \PC\ of 
			$\la N[y]\ra$ has colors $c_2$ and $c_4$. Since on edge $(x,y)$, resp. 
			$(x,1)$, both colors $c_1$ and $c_2$, resp. $c_3$ and $c_4$ are 
			represented we can identify those colors and merge them to one color, 
			resulting in a proper combined coloring.
		   Hence, the \PC\ $P_{\la N[x]\ra}$ 	is a color-continuation of
			$P_{\la N[y]\ra}$ and vice versa.			}
\label{fig:exmplCC}
\end{figure}

In other words, for all newly colored edges with color $c$ in $H_2$, 
which are Cartesian edges in $H_2$ that satisfy the \Scond\ in $H_2$, 
we have to find a representative edge that satisfy the \Scond\ in $H_1$ and 
was already colored in $H_1$. If $H_1$ and $H_2$ are thin 
we can ignore the \Scond, since all edges satisfy this 
condition in $H_1$ and $H_2$, see Figure \ref{fig:exmplCC}.

However, there are cases where the color-continuation fails, see
Figure \ref{fig:colorcontifails}. The remaining part of this subsection
is organized as follows. We first show how one can solve the 
color-continuation problem if the corresponding subproducts are thin.
As it turns out, it is sufficient to use the information
of 1-neighborhoods only in order to get a proper
combined coloring. We then proceed to solve this problem
for non-thin subgraphs. 

Before we continue, two important lemmas are given.
The first one is just a restatement of a lemma, which was formulated for 
equivalence classes w.r.t. to a product relation in \cite{IMZE-94}.
The second lemma shows how one can adapt this lemma to
non-thin graphs.

\begin{lem}[\cite{IMZE-94}, Lemma 1]
	Let $G$ be a thin strong product graph 
	and let $P_{G}$  be a \PC\ of $G$.
	Then every vertex of $V(G)$ is incident to at least one 
	edge with color $c$ for all colors 
	$c$ in the image of $P_{G}$.
	\label{lem:every_color_on_every_vertex}
\end{lem}

\begin{lem}
	Let $G$ be a thin strong product graph, $H\subseteq G$
	be a non-thin subproduct of $G$ and $x\in V(H)$ be
	a vertex with $|S_H(x)|=1$. 
	Moreover, let $P_{H}$ be a partial \PC\ of $H$. 
	Then vertex $x$  is contained in at least one edge with color $c$ 
	for all colors 	$c$ in the image of $P_{G}$.
	\label{lem:every_partialcolor_on_every_vertex}
\end{lem}

\begin{proof}
	Notice that $H$ does not contain complete
	factors, otherwise Corollary \ref{cor:common_cardi_of_S} implies that $|S_H(x)|>1$.
	Now, the statement follows directly from Lemma \ref{lem:S=1condition} 
	and Lemma \ref{lem:every_color_on_every_vertex}
\end{proof}

\subsubsection{Solving the Color-Continuation Problem for Thin Subgraphs}

To solve the color-continuation problem for
thin subgraphs and in particular for thin 1-neighborhoods
we introduce so-called \emph{S-prime} graphs \cite{Hel-13,Sab:75,LB:81, LB:95,Klavzar:02,Bresar:03,HGS-09}.

\begin{defn}
A graph $S$ is \emph{S-prime} (\emph{S} stands for ``subgraph'')
if for all graphs $G$ and $H$ with $S\subseteq G\star H$ holds: $S\subseteq
H$ or $S\subseteq G$, where $\star$ denotes an arbitrary graph product.   
\end{defn}

The class of S-prime graphs was introduced and characterized 
for the direct product by Sabidussi in 1975 \cite{Sab:75}. 
Analogous notions of S-prime graphs with respect to other
products are due to Lamprey and Barnes \cite{LB:81, LB:95}.
Klav{\v{z}}ar \emph{et al.} \cite{Klavzar:02} 
and Bre{\v{s}}ar \cite{Bresar:03} proved several
characterizations of (basic) S-prime graphs.
In \cite{HGS-09} it is shown that so-called 
\emph{diagonalized Cartesian products} 
of S-prime  graphs are S-prime w.r.t. the Cartesian product.
We shortly summarize the results of \cite{HGS-09}.

\begin{defn}[\cite{HGS-09}]
  A graph $G$ is called a \emph{diagonalized} Cartesian product, 
  whenever there is an edge $(u,v) \in E(G)$ such that 
  $H = G\setminus(u,v)$ 
  is a nontrivial Cartesian product and $u$
  and $v$ have maximal distance in $H$. 
  \label{def:diag_CartProd}
\end{defn}

\begin{thm}[\cite{HGS-09}]\label{thm:main_result}
  The diagonalized Cartesian Product of S-prime graphs is 
  S-prime w.r.t. the Cartesian product.
\end{thm}

\begin{cor}[\cite{HGS-09}]
  Diagonalized Hamming graphs, and thus diagonalized Hypercubes, are
  S-prime w.r.t. the Cartesian product.
	\label{cor:diagH}
\end{cor}

We shortly explain how S-prime graphs can be used in 
order to obtain a proper color-continuation
in thin subproducts even if the color-continuation
fails. Consider a strong product graph $G$ and two given thin 
subproducts $H_1,H_2 \subseteq G$. 
Let the Cartesian edges of each subgraph 
be colored  with respect to a product coloring 
of $H_1$, respectively $H_2$ that is at least
as fine as the  product coloring of $G$ w.r.t. to its PFD. 
As stated in Definition \ref{continuation}, 
we have a proper color-continuation from $H_1$  to $H_2$
if for all colored edges with color $c$ in $H_2$
there is a representative edge that is colored in $H_1$.
Assume the color-continuation fails, i.e., 
there is a color $c$ in $H_2$ such that
for all edges $e_c \in E(H_2)$ with color $c$
holds that $e_c$ is not colored in $H_1$, for an example
see Figure  \ref{fig:colorcontifails}. 
This implies that all such edges $e_c$ are determined as non-Cartesian in
$H_1$.
As claimed,  the product colorings of 
$H_1$ and $H_2$ are at least as fine as the one of $G$ and 
$H_1$, $H_2$ are subproducts of $G$, 
which implies that colored Cartesian edges 
in each $H_i$ are Cartesian edges in $G$.
Since $e_c$ is determined as non-Cartesian in $H_1$, 
but as Cartesian in $H_2$, 
we can infer that $e_c$ must be Cartesian in $G$. 
Thus we can force the edge $e_c$ 
to be Cartesian in $H_1$.
The now arising questions is:  "What happens with the factorization of $H_1$?" 
We will show in the sequel that there is a hypercube in
$H_1$ consisting of Cartesian edges only, where all edges are copies
of edges of different factors. Furthermore, we show that this hypercube
 is diagonalized by a particular edge $e_c$ and therefore S-prime w.r.t
the Cartesian product.
Moreover, we will prove that all colors that appear
on this hypercube and the color $c$ on $e_c$
have to be merged to exactly one color, even with respect
to the product coloring, provided by the coloring
w.r.t. the strong product.
This approach solves the color-continuation problem
for thin subproducts and hence in particular for
thin 1-neighborhoods as well.

\begin{figure}[tb] 
  \centering
   \includegraphics[bb = 125 470 417 700, scale=1.1]{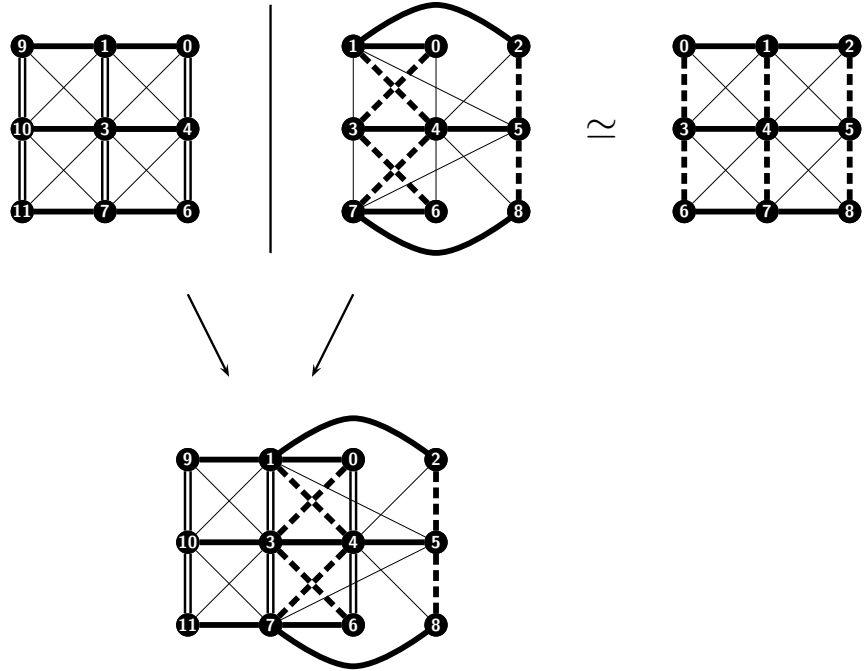}
  \caption[Example: color-continuation fails]{ \textit{Color-continuation problem in thin subproducts.} 
			Consider the induced neighborhoods $\la N[3] \ra$ and $\la N[4] \ra$,
			depicted in the upper part. The colorings of the edges 
			w.r.t. the PFD of each neighborhood are shown as 
			thick dashed edges, thick-lined edges and double-lined edges, respectively.
			If we cover the graph $G$ in the lower part from $N[3]$ to $N[4]$
			the color-continuation fails, e.g. on edge $(1,4)$,
			since $(1,4)$ is determined as non-Cartesian in $\la N[3] \ra$.
			This holds for all edges in $\la N[3] \ra$ that obtained 
			the color "thick dash" in $\la N[3] \ra$. 
			The same holds for the color	
			"double-lined" if we cover the graph from $N[4]$ to $N[3]$.
			If we force the edge $(1,4)$ to be Cartesian in $\la N[3] \ra$ 
			Lemma \ref{lem-Sprime-addCartesian} implies that the colors
			"thick-lined" and "double-lined" have to be merged to one color,
			since the subgraph with edge set 
			$\{(0,1),(0,4),(1,3),(3,4)\} \cup \{(1,4)\}$ 
			 is a diagonalized hypercube $Q_2$. Note,  $G$ can be covered by thin
			1-neighborhoods only, but the color-continuation fails. Hence	
			 $G$ is not NICE in the terminology of \cite{HIKS-08}.}
  \label{fig:colorcontifails}
\end{figure}

\begin{lem}
	Let $G = \boxtimes_{l=1}^n G_l$ be a thin strong product graph and 
	$(v,w)\in E(G)$ a non-Cartesian edge. Let
    $J$ denote the set of indices where $v$  
	and $w$ differ and $U\subseteq V(G)$ 
	be the set of vertices $u$ with coordinates  $u_i= v_i$, if $i\notin J$ 
	and $u_i \in\{v_i,w_i\}$, if $i\in J$.
	Then the induced subgraph $\la U\ra \subseteq \skel(G)$ on $U$
	consisting of Cartesian edges of $G$ only is a hypercube of dimension $|J|$.
	\label{lem-induHypercube}
\end{lem}

\begin{proof}
	Notice that the coordinatization of $G$ is unique, since $G$ is thin.
	Moreover, since the strong product is commutative and associative we can
	assume w.l.o.g. that $J = \{1,\dots,k\}$. Note, that $k>1$, otherwise
	the edge $(v,w)$ would be Cartesian. 

	Assume that $k=2$. We denote the coordinates of $v$, resp. of $w$,
	by $(v_1,v_2,X)$, resp. by $(w_1,w_2,X)$. By definition of
	the strong product we can conclude that $(v_i,w_i) \in E(G_i)$ for
	$i =1,2$. Thus the set of vertices with coordinates $(v_1,v_2,X)$
	$(v_1,w_2,X)$,$(w_1,v_2,X)$, and $(w_1,w_2,X)$ induce a complete
	graph $K_4$ in $G$. Clearly, the subgraph consisting of 
	Cartesian edges only is a $Q_2$.

	Assume now the assumption is true for $k=m$. We have to show
	that the statement holds also for $k=m+1$. Let J=\{1,\dots,m+1\} and
	let $U_1$ and $U_2$ be a partition of $U$ with
	$U_1 = \{u\in U \mid u_{m+1} = v_{m+1}\}$ and 
	$U_2 = \{u\in U \mid u_{m+1} = w_{m+1}\}$.
    Thus each $U_i$ consists of vertices that differ only
	in the first $m$ coordinates. Notice, by 
	definition of the strong product and by construction of both sets $U_1$
	and $U_2$ there are vertices $a,b$ in each $U_i$  that differ
	in all $m$ coordinates that are adjacent in $G$ and hence non-Cartesian
	in $G$. Thus, by induction hypothesis
	the subgraphs $\la U_i \ra$ induced by each $U_i$ consisting
	of Cartesian edges only is a $Q_m$.
	Let $\la U\ra$ be the subgraph with  vertex set $U$ and edge set
	$E(\la U_1\ra)\cup E(\la U_2\ra)\cup \{(a,b)\in E(G) \mid a = (X,v_{m+1},Y)$ and $b = (X,w_{m+1},Y)\}$.
	By  definition of the strong product 
	the edges $(a,b)$ with $a = (X,v_{m+1},Y)$ and $b = (X,w_{m+1},Y)$ 
	induce	an isomorphism between $\la U_1 \ra$ and $\la U_2\ra$ which implies 
    that $\la U \ra \simeq Q_m \Box K_2 \simeq Q_{m+1}$.
\end{proof}

\begin{lem}
	Let $G  = \boxtimes_{l=1}^n G_l $ be a thin strong product graph, 
	where each $G_l$, $l=1,\dots,n$ is prime. 
	Let 
	$H = \boxtimes_{l=1}^m H_l \subseteq G$ be a 
	thin subproduct of $G$ such that there is
	a non-Cartesian edge $(v,w)\in E(H)$ that is Cartesian in $G$. 
	Let $J$ denote the set of indices where $v$  
	and $w$ differ w.r.t. the coordinatization of $H$.
	Then the factor $\boxtimes_{i\in J} H_i$ of $H$ is a subgraph
	of a prime factor $G_l$ of $G$. 
	\label{lem-Sprime-addCartesian}
\end{lem}

\begin{proof}
	In this proof, factors w.r.t. the Cartesian product and 
	the strong product, respectively, are called Cartesian factors and 
	strong factors, respectively.
	First notice that Cartesian edges in $G$ as well as in $H$
	are uniquely determined, since both graphs are thin.
 	Moreover, the existence of a Cartesian edge of $G =\boxtimes_{l=1}^n G_l$, 
	that is a non-Cartesian edge in a subproduct $H = \boxtimes_{l=1}^m H_l$
	of $G$, implies that $m>n$, i.e., the factorization of $H$ is
	a refinement of the factorization induced by the global PFD.
	Since $H$ is a thin subproduct of $G$ 
	with a refined factorization, it follows that Cartesian
	edges of $H$ are Cartesian edges of $G$. Therefore,  
	we can conclude that strong factors of $H$ are entirely contained
	in strong factors of $G$.

	We denote the subgraph of $H$ that consists of all
	Cartesian edges of $H$ only, i.e., its Cartesian skeleton, by $\skel(H)$, 
	hence $\skel(H) =  \Box_{l=1}^m H_l$. 
	Let  $U\subseteq V(H)$ be the set of vertices 
	$u$ with coordinates  $u_i= v_i$, if $i\notin J$ 
	and $u_i \in\{v_i,w_i\}$, if $i\in J$. 
	Notice that Lemma \ref{lem-induHypercube} implies that for the induced
	subgraph w.r.t. the Cartesian skeleton 
	$\la U \ra \subseteq \skel(H)$ holds $\la U \ra\simeq Q_{|J|}$. 
	Moreover, the distance $d_{\la U \ra}(v,w)$ between 
	$v$ and $w$ in $\la U \ra$ is $|J|$, that is the maximal
	distance that two vertices can have in $\la U \ra$. If we claim that
	$(v,w)$ has to be an edge in $\la U \ra$ we obtain 
	a diagonalized hypercube $\la U \ra^{diag}$.
	Corollary \ref{cor:diagH} implies that $\la U \ra^{diag}$ 
	is S-prime and hence $\la U \ra^{diag}$ must be contained 
	entirely in a Cartesian factor $\widetilde H$ of 
	a graph $H^* = \widetilde H \Box H'$ with $\skel(H) \cup (v,w)\subset H^*$. 
	This implies that $\la U \ra^{diag} \subseteq \widetilde H^u$
   	for all $u\in V(H^*)$, i.e., $\la U \ra^{diag}$ is entirely contained
	in all $\widetilde H^u$-layer in $H^*$. Note that all
	$\widetilde H$-layer $\widetilde H^u$ contain at least one edge of every $H_i$-layer $H_i^u$
	of the previously determined factors $H_i$, $i\in J$ of $H$.

	Furthermore, all Cartesian factors 
	of $\skel(H) = \Box_{l=1}^m H_l$ coincide with the strong factors of 
	$H = \boxtimes_{l=1}^m H_l$ and hence, in particular the factors
	$H_i,\ i\in J$.
	Moreover, since $H$ is a subproduct of $G$ and 
	the factorization of $H$ is a refinement of $G$ it holds that
    Cartesian factors $H_i,\ i\in J$  of $\skel(H)$  must be entirely contained 
	in strong prime factors of $G$.	This implies 
	that for all $i\in J$ the $H_i$-layer $H_i^u $ must be 
	entirely contained in the layer of strong factors of $G$.
	We denote the set of all already determined strong factors
	$H_i,\ i\in J$ of $H$ with $\mathcal H$.

	Assume the graph $H^* = \Box_{j=1}^s K_j$ with $\skel(H) \cup (v,w) \subseteq H^*$
	and $V(H^*) = V(\skel(H))$ 
	has a factorization such that $\Box_{i\in J} H_i \cup (v,w) \not \subseteq K_j$
	for all Cartesian factors $K_j$. Since $\skel(H) \cup (v,w) \subseteq H^*$, 
	we can conclude that $\la U \ra^{diag} \subseteq H^*$. Since 
	$\la U \ra^{diag}$ is S-prime it must be contained in
	a Cartesian factor $K_r$ of $H^*$. This implies
	that $\la U \ra^{diag}\subseteq K_r^u$ for all $u\in V(H^*)$, i.e.,
	for all $K_r$-layer of this particular Cartesian factor $K_r$. 
	Since $\Box_{i\in J} H_i \cup (v,w) \not \subseteq K_r$, we can conclude
	that there is an already determined strong factor $H_i$
	such that $H_i^u \not\subseteq K_r^u$ for all $u\in V(H^*)$.
	Furthermore, all $K_r$-layer $K_r^u$ contain at least one edge of each $H_i$-layer $H_i^u$
	of the previously determined strong factors $H_i$, $i\in J$ of $H$.
	We denote with $e$ the edge of the $H_i$-layer $H_i^u$ that is contained 
	in the $K_r$-layer $K_r^u$.
	This edge $e$ cannot be contained in any $K_j$-layer, $j\neq r$.
	This implies that $H_i^u \not\subseteq K_j^u$ for any $K_j$-layer,
	$j=1,\dots,s$.

	Thus,	there is an already determined strong factor 
	$H_i \in \mathcal H$ with $H_i^u \not\subseteq K_j^u,\ u\in V(H^*)$ 
	for all $K_j$-layer in $H^*$, $j=1,\dots,s$. Therefore, none
	of the layer of this particular $H_i$ are subgraphs of layer 
	of any Cartesian factor $K_j$ of $H^*$. 
	This means that $H^*$ is not a subproduct of $G$
	or a refinement of $H$, both cases contradict that $H_i \in \mathcal H$. 
	
	Therefore, we can conclude that 
	$\la U \ra^{diag} \subseteq \Box_{i\in J} H_i \cup (v,w) \subseteq \widetilde H$
	for a Cartesian factor $\widetilde H$ of $H^*$. As argued,
	Cartesian factors are subgraphs of its strong factors 
	and hence, we can infer that $\Box_{i\in J} H_i$
	and hence  $\boxtimes_{i\in J} H_i$ 
	must be entirely contained  in a strong factor
	of $H$ and hence in a strong factor of $G$,
	since $H$ is a subproduct.
\end{proof}

\subsubsection{Solving the Color-Continuation Problem for Non-Thin Subgraphs}

The disadvantage of non-thin subgraphs is that, in contrast
to thin subgraphs, not all edges satisfy the \Scond.
The main obstacle is that the color-continuation can fail
if a particular color is represented on edges that
don't satisfy the \Scond\ in any used subgraphs.  
Hence, those edges cannot be identified as Cartesian
in the corresponding subgraphs, see Figure \ref{fig:exmplCC2}. 
Moreover, we cannot
apply the approach that is developed for thin subgraphs
by usage of diagonalized hypercubes in general.
Therefore, we will extend 1-neighborhoods and use
also edge- and $N^	*$-neighborhoods.

\begin{figure}[htbp]
  \centering
  \includegraphics[bb= 126 514 441 744, scale=1]{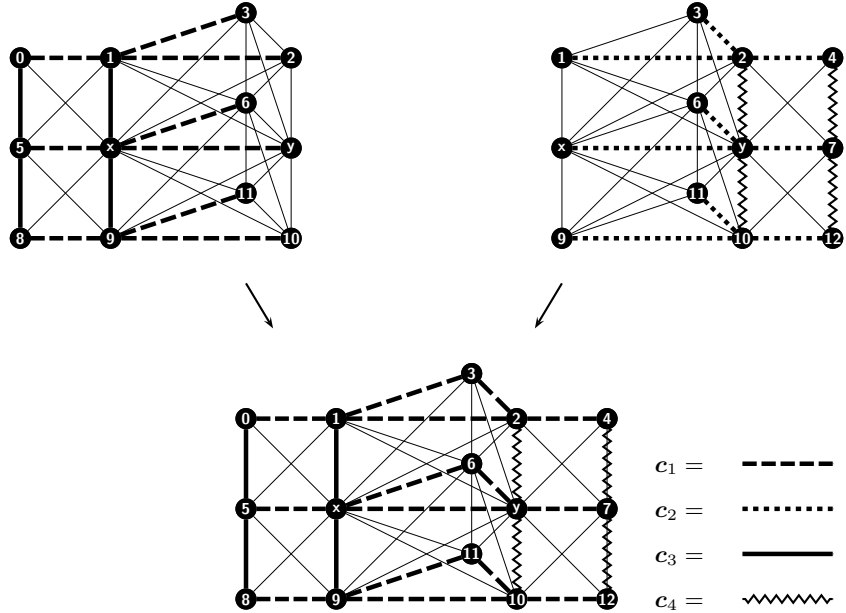}
  \caption[Example: color-continuation fails]{ \textit{Color-continuation problem in non-thin subproducts.} 
			Shown is a thin graph $G$ that is a strong product of
		 	a path and a path containing a triangle. The backbone $\bone(G)$
			consists of the vertices $x$ and $y$. Both neighborhoods
			$\la N[x]\ra$ and $\la N[y]\ra$ are not thin. 
			After computing the PFD of $\la N[x]\ra$, resp. of $\la N[y]\ra$ 
			one obtains a partial product coloring with colors $c_1$ and $c_3$,
			resp. 	with colors $c_2$ and $c_4$. In this example 
			the partial product coloring of $P_{\la N[y]\ra}$ is not
			a color-continuation of $P_{\la N[x]\ra}$ since no
			edge with color $c_4$ is colored in $\la N[x]\ra$.
			}
\label{fig:exmplCC2}
\end{figure}

In the following, we will provide several properties 
of (partial) product colorings and show 
that in a given thin strong product graph $G$
a partial \PC\ $P_{H}$ of a subproduct
$H \subseteq G$ is always a color-continuation of a 
partial \PC\ $P_{\langle N[x]\rangle}$
of any 1-neighborhood $N[x]$  with $N[x]\subseteq V(H)$ and $x \in \bone(G)$ and vice versa.
This in turn implies that we always get a proper color-continuation
from any 1-neighborhood $N[x]$ to edge-neighborhoods of edges $(x,y)$
and to $N^*_{x,y}$-neighborhoods with $y\in N[x]$ and vice versa.

\begin{lem}
	Let $G$ be a thin graph and $x\in \bone(G)$.
	Moreover let $P^1$ and $P^2$ be arbitrary partial
	\PC s of the induced neighborhood $\langle N[x] \rangle$.

	Then $P^2$ is a color-continuation of $P^1$ and vice versa.
	\label{lem:colorconti_itself}
\end{lem}

\begin{proof}
	Let $C^1$ and $C^2$ denote the images of $P^1$ and $P^2$, 
	respectively.
	Note, the PFD of $\langle N[x] \rangle$ is the finest
	possible factorization, i.e., the number of used colors
	becomes maximal. Moreover, every fiber with respect to the 
	PFD of $\langle N[x] \rangle$ that satisfies the \Scond, 
	is contained in any decomposition of $\langle N[x] \rangle$.
	In other words any prime fiber that satisfies the \Scond\  
	is a subset of a fiber that satisfies the \Scond\ with respect to any 
	decomposition of $\langle N[x] \rangle$.

	Moreover since	$x\in\bone(G)$ it holds that $|S_x(x)|=1$ and thus
	every edge containing vertex $x$ satisfies the \Scond\ 
	in $\langle N[x] \rangle$. 
	Lemma \ref{lem:S=1condition} implies that all Cartesian edges $(x,v)$
	can be determined as Cartesian in $\langle N[x] \rangle$.  
	Together with Lemma	\ref{lem:every_partialcolor_on_every_vertex}
	we can infer that each color of $C^1$, resp. $C^2$
	is represented at least on edges $(x,v)$ contained in the 
	prime fibers, which completes the proof.  
\end{proof}

\begin{lem}
	Let $G = \boxtimes_{i=1}^n G_i$ be a thin strong product graph. 
	Furthermore	let $H$ be a subproduct of $G$ with partial \PC\  $P_{H}$
	and $\langle N[x] \rangle \subseteq H$ with $x\in \bone(G)$. 
	
	Then $P_H$ is a color-continuation of the partial \PC\  
	$P_N$  of $\langle N[x]\rangle$ and vice versa.
	\label{lem:color-conti_NtoH}
\end{lem}

\begin{proof}
	First notice that Lemma \ref{lem:s1_in_H} implies that $x\in \bone(H)$ and
	in particular $|S_H(x)|=1$.
	Thus  every edge containing 
	vertex $x$ satisfies the \Scond\ in $H$ as well as in $\langle N[x] \rangle$.
	Moreover, Lemma \ref{lem:every_partialcolor_on_every_vertex}
	implies that every color of the partial  \PC\ $P_H$, resp. $P_N$,
	is represented at least on edges $(x,v)$. 

	Since $\langle N[x] \rangle$ is a subproduct
	of the subproduct $H$ of $G$ we can conclude that
	the PFD of $H$ induces a local (not necessarily prime)	
	decomposition of $\langle N[x] \rangle$ and hence
	a partial \PC\ of $\la N[x]\ra$.
	Lemma \ref{lem:colorconti_itself} implies that
	any partial \PC\  of $\langle N[x] \rangle$ and hence
	in particular the one induced by $P_H$ is 
	a color-continuation of $P_N$.

	Conversely, any product coloring $P_N$ of 
	$\langle N[x] \rangle$ is a color-continuation 
	of the product coloring induced by the PFD of 
	$\langle N[x] \rangle$. Since $\langle N[x] \rangle$
	is a subproduct of $H$ it follows that every prime
	fiber of  $\langle N[x] \rangle$ that satisfies the \Scond\  
	is a subset of a prime fiber of $H$ 
	that satisfies the \Scond. 
	This holds in particular for the fibers through
	vertex $x$, since $|S_x(x)|=1$ and $|S_H(x)|=1$. 
	By the same arguments as in the proof of Lemma
	\ref{lem:colorconti_itself} one can infer that every product
	coloring of $H$ is a color-continuation of 
	the product coloring induced by the PFD of $H$,
	which completes the proof.
\end{proof}

We can infer now the following Corollaries.

\begin{cor}
	Let $G = \boxtimes_{i=1}^n G_i$ be a thin strong product graph,
	$(v,w) \in E(G)$ be a Cartesian edge of $G$ and $H$ denote
	the edge-neighborhood $\langle N[v]\cup N[w] \rangle$. 
	Then any partial \PC\  $P_{H}$ of $H$ is a color-continuation
	of any partial \PC\ $P_{N[v]}$ of $\langle N[v] \rangle$,
	resp.  of any partial \PC\ $P_{N[w]}$ of $\langle N[w] \rangle$
	and vice versa.
	\label{cor:color-conti_edgeN}
\end{cor}

\begin{cor}
	Let $G = \boxtimes_{i=1}^n G_i$ be a thin strong product graph and 
	$(v,w) \in E(G)$ be an arbitrary edge of $G$.
	Then any partial \PC\  $P^*$ of the $N^*_{v,w}$-neighborhood 
	is a color-continuation	of any partial \PC\ $P_{N[v]}$ of $\langle N[v] \rangle$,
	resp.  of any partial \PC\ $P_{N[w]}$ of $\langle N[w] \rangle$
	and vice versa.
	\label{cor:color-conti_N*}
\end{cor}

\section{A Local PFD Algorithm for Strong Product Graphs}

In this section, we use the previous 
results and provide a general local approach for the PFD
of thin graphs $G$. Notice that even if 
the given graph $G$ is not thin, the provided 
Algorithm works on $G/S$. 
The prime factors of $G$ can then be constructed
by using the information of the prime factors of $G/S$
by repeated application of Lemma 5.40 provided in \cite{IMKL-00}.

In this new PFD approach we use in addition an
algorithm, called \emph{breadth-first search (BFS)},
that traverses all vertices of a graph $G=(V,E)$ in a particular order.
We introduce the ordering of the vertices of $V$ 
by means of breadth-first search as follows:
Select an arbitrary vertex $v\in V$ and create a sorted list $BFS(v)$
of vertices beginning with $v$; append all neighbors $v_1,\ldots,v_{\deg(v)}$ of
$v$; then append all neighbors of $v_1$ that are not already in this list; 
continue recursively with
$v_2,v_3,\ldots$ until all vertices of $V$ are processed.
In this way, we build levels 
where each $v$ in level $i$ is adjacent to some vertex $w$ in
level $i-1$ and vertices $u$ in level $i+1$.
We then call the vertex $w$ the \emph{parent} of $v$
 and vertex $v$ a \emph{child} of $w$.

We give now an overview of the new approach. Its top level control 
structure is summarized in Algorithm \ref{alg:general}.

Given an arbitrary thin graph $G$, first
the backbone vertices are ordered via the 
\emph{breadth-first search (BFS)}.
After this, the neighborhood of the first 
vertex $x$ from the ordered BFS-list $\bone_{BFS}$ is decomposed.
Then the next vertex $y \in N[x] \cap \bone_{BFS}$ is taken 
and the edges of $\langle N[y] \rangle$ are colored with 
respect to the neighborhoods PFD.
If the color-continuation from $\la N[x] \ra$ to $\la N[y] \ra$ does not fail, then the 
Algorithm  proceeds with the next vertex $y' \in N[x] \cap \bone_{BFS}$.
If the color-continuation fails and both neighborhoods are thin, 
one uses 
Algorithm \ref{alg:checkHypercube} in order to compute 
a proper combined coloring. If one of the neighborhoods
is non-thin the Algorithm proceeds 
with the edge-neighborhood $\langle N[x]\cup N[y] \rangle$.
If it turns out that $(x,y)$ is indispensable in $\langle N[x]\cup N[y] \rangle$
and hence, that $\langle N[x]\cup N[y] \rangle$ is a 
proper subproduct (Corollary \ref{cor:proper_edgeN})
 the algorithm proceeds to decompose and to color
$\langle N[x]\cup N[y] \rangle$.
If it turns out that $(x,y)$ is dispensable in $\langle N[x]\cup N[y] \rangle$
the $N^*$-neighborhoods $N^*_{x,y}$ is factorized and colored.
In all previous steps edges are marked as "checked" if they
satisfy the \Scond, independent from being Cartesian or not.
 After this, the $N^*$-neighborhoods of 
all edges that do not satisfy the
\Scond\ in any of the previously used subproducts, i.e,
1-neighborhoods, edge-neighborhoods or $N^*$-neighborhoods,
are decomposed and again the edges are colored.
Examples of this approach are depicted in
Figure \ref{fig:Exmpl_gen1} and \ref{fig:Counter_edgeN_gen}.
Finally, the Algorithm checks which of the recognized factors have 
to be merged into the prime factors $G_1,\dots,G_n$ of $G$.

\begin{algorithm}[ht]
\caption{General Approach}
\label{alg:general}
\begin{algorithmic}[1]
\renewcommand{\baselinestretch}{0.9}
\small\normalsize
\vspace{1mm}
    \STATE \textbf{INPUT:} a thin graph $G$
	\STATE compute backbone-vertices of $G$, order them in BFS and store them in $\bone_{BFS}$;
	\STATE $x \gets$ first vertex of $\bone_{BFS}$;	
	\STATE $W \gets \{x\}$;	
	\STATE FactorSubgraph($\langle N[x] \rangle$);
    \WHILE{$\bone_{BFS} \neq\emptyset$}\label{b:Part1-start}
		\STATE $H \gets \la \cup_{w\in W} N[w] \ra$; \label{b:H1}
    	\FOR {all $y \in N[x] \cap \bone_{BFS}$}\label{forLoop}
        	  \STATE FactorSubgraph($\langle N[y] \rangle$);
			  \STATE compute the combined coloring of $H$ and $\langle N[y] \rangle$;		
    	      \IF{color-continuation fails from $H$ to $N[y]$}    
					\IF{$\la N[x]\ra$ and $\la N[y]\ra$ are thin}
						\STATE $C \gets \{\text{color }c \mid \text{color-continuation for }c\text{ fails}\}$;
						\STATE Solve-Color-Continuation-Problem(H, $\langle N[y]\rangle$, x, C);
								\COMMENT{Algorithm \ref{alg:checkHypercube}}
   				        \STATE  mark all vertices and all edges of $\langle N[y]\rangle$ as "checked";
					\ELSIF{(x,y) is indispensable in $\langle N[x]\cup N[y]\rangle$}
						\STATE  FactorSubgraph($\langle N[x]\cup N[y]\rangle$);
					\ELSE  \STATE	FactorSubgraph($N^*_{x,y}$);
					\ENDIF		
					\STATE compute the combined coloring of $H$ and $\langle N[y] \rangle$;
					\label{line-comb}			  
			  \ENDIF			  
		\ENDFOR	
		\STATE  delete $x$ from $\bone_{BFS}$;
		\STATE  $x \gets$ first vertex of $\bone_{BFS}$;
		\STATE $W \gets W \cup \{x \}$;
	\ENDWHILE\label{b:Part1-end}
	\WHILE{there exists a vertex $x \in V(H)$ that is not marked as "checked"} 
			\IF{there exist edges $(x,y)$ that are not marked as "checked"} 
				\STATE FactorSubgraph($N^*_{x,y}$);
			\ELSE \STATE take an arbitrary edge $(x,y) \in E(H)$;
				  \STATE FactorSubgraph($N^*_{x,y}$);
			\ENDIF
			\STATE compute the combined coloring of $H$ and $N^*_{x,y}$;
	\ENDWHILE
	\FOR{each edge $e\in E(H)$}\label{b:lastFor1}
		\STATE assign color of $e$ to edge $e \in E(G)$;
	\ENDFOR \label{b:lastFor2}
	\STATE CheckFactors(G); \COMMENT{check and merge factors with Algorithm \ref{alg:iso-test}\label{b:iso}}
	\label{isomorph-test-end}
    \STATE \textbf{OUTPUT:} G with colored $G_j$-fiber, and Factors of $G$;
\renewcommand{\baselinestretch}{1.3}
\small\normalsize
\end{algorithmic}
\end{algorithm}

Before we proceed to prove the correctness of this local PFD algorithm, 
we show that we always get a proper combined coloring
by usage of Algorithm \ref{alg:checkHypercube}.

\begin{lem}
Let $G$ be a thin graph and $\bone_{BFS} = \{v_1,\dots,v_n\}$ be its 
	BFS-ordered sequence of backbone vertices. 
	Furthermore, let $H=\la \cup_{j=1}^{i-1} N[v_j] \ra$
	be a partial product colored subgraph of $G$ that obtained
	its coloring from a proper combined product coloring induced by the PFD w.r.t. 
    the strong product of each
	$\la N[v_j] \ra$, $j = 1,\dots,i-1$. 
    Let $\la N[v_i] \ra$
	be a thin neighborhood that is product colored w.r.t. to its PFD. 
	Let vertex $x$ denote the parent of $v_i$.  
	Assume $\la N[x] \ra$ is thin.
	Moreover, assume the color-continuation from $H$ to $\la N[v_i] \ra$
	fails and let $C$ denote the set of colors where it fails.

	Then Algorithm \ref{alg:checkHypercube} computes 
	a proper combined coloring of the colorings of $H$ and $\la N[v_i] \ra$
	with  $H$, $\la N[v_i] \ra$, $x$ and $C$ as input.
	\label{lem:properCC-via-sprime}
\end{lem}

\begin{proof}
	First notice that it holds $\la N[x]\ra \subseteq H=\la \cup_{j=1}^{i-1} N[v_j] \ra$.  
	Let $c \in C$. Hence, $c$ denotes a color in $\langle N[v_i] \rangle$ such that
	for all edges $e \in E(\langle N[v_i] \rangle)$ with color $c$
	holds that $e$ was not colored in $H$.	
	Since the combined coloring in $H$ implies a 
	product coloring of $\langle N[x] \rangle$ we
	can compute the coordinates of the vertices 
	in $\langle N[x] \rangle$ with respect 
	to this coloring. Notice that the coordinatization
	in $\langle N[x] \rangle$ is unique since 
	$\langle N[x] \rangle$ is thin.
	Now Lemma \ref{lem:every_color_on_every_vertex} 
	implies that there 
	is at least one edge $e\in \la N[v_i]\ra$ with color $c$ that
	contains vertex $x$, since $x \in N[v_i]$. Let us denote this edge
	by $e_c = (x,w)$. Clearly, it holds $(x,w)\in E(\langle N[x] \rangle)$.
	Hence, this edge is not determined as Cartesian in $H$,
	and thus in particular not in $\langle N[x] \rangle$
	otherwise $e_c$ would have been colored in $\langle N[x] \rangle$. 
	But since $e_c$ is determined as Cartesian in $\langle N[v_i] \rangle$
	and moreover, since $\langle N[v_i] \rangle$ is a subproduct of $G$, we can infer
	that $e_c$ must be Cartesian in $G$. 
	Therefore, we claim that the non-Cartesian edge
	$(x,w)$ in $\langle N[x] \rangle$ has to be Cartesian in
	$\langle N[x] \rangle$. Notice that the product coloring
	of $\langle N[x] \rangle$ induced by the combined colorings
	of all $\la N[v_j]\ra$, $j=1,\dots,i-1$ is as least
	as fine as the product coloring of $G$. Thus, we can
	apply Lemma \ref{lem-Sprime-addCartesian}  and together with 
    the unique coordinatization of $\langle N[x] \rangle$
	 directly conclude that all colors $k \in D$,
	where $D$ denotes the set of coordinates
	 where $x$ and $w$ differ, have to be merged to one color. 
	This implies that we always get a proper combined coloring and hence
	a proper color-continuation for each such color $c$ that is 
	based on those additional edges $e_c = (x,w)$ as defined above.
\end{proof}

\begin{algorithm}[tbp]
\caption{Solve-Color-Continuation-Problem}
\label{alg:checkHypercube}
\begin{algorithmic}[1]
\renewcommand{\baselinestretch}{0.9} 
\small\normalsize
\vspace{1mm}
    \STATE \textbf{INPUT:} a partial product colored graph $H$, 
						   a product colored graph $\langle N[v_i]\rangle$, 
							a vertex $v$, set $C$ of colors 
				\STATE compute coordinates of $\langle N[v]\rangle$ with respect
					   to the combined product coloring of $H$; 
				\STATE \COMMENT{color "j" if differ in coordinate "j"}
				\FOR{all colors $c \in C$  \COMMENT{color-continuation fails}}
					\STATE take one representative $e_c = (v,w)\in E(\langle N[v_i]\rangle)$;
					\STATE $D\gets \{k \mid v \text{ and } w \text{ differ in coordinate } k \}$;
					\STATE merge all colors $k\in D$ in $H$ to one color;
				\ENDFOR
			\STATE compute the combined coloring of $H$  and $\langle N[v_i] \rangle$;
   \STATE \textbf{OUTPUT:} colored graph H, colored graph $\langle N[v_i] \rangle$;
\renewcommand{\baselinestretch}{1.3} 
\small\normalsize

\end{algorithmic}
\end{algorithm}	

\begin{algorithm}[tpb]
\caption{FactorSubgraph}
\label{alg:factorSub}
\begin{algorithmic}[1]
\renewcommand{\baselinestretch}{0.9}
\small\normalsize
\vspace{1mm}
    \STATE \textbf{INPUT:} a graph $H$
		\STATE compute the PFD of $H$ and color
    	   	   the Cartesian edges in $H$ that satisfy the \Scond;
		\STATE  mark all vertices $x$ with $|S_H(x)| = 1$ as "checked";
		\STATE  mark all edges that satisfy the \Scond\ as "checked";
	\STATE \textbf{Return} partially colored $H$;
\renewcommand{\baselinestretch}{1.3}
\small\normalsize	
  \end{algorithmic}
\end{algorithm}

\begin{algorithm}[tp]
\caption{CheckFactors}
\label{alg:iso-test}
\begin{algorithmic}[1]
\renewcommand{\baselinestretch}{0.9}
\small\normalsize
\vspace{1mm}
	\STATE \textbf{INPUT:} a thin product colored graph $G$
	\STATE take one connected component $G^*_1,\dots,G^*_l$ of each color $1,\dots,l$  in $G$; 
	\STATE $I \gets  \{1,\dots,l\} $;
    \STATE $J \gets I$;
    \FOR{$k=1$ to $l$}
      \FOR{each $S \subset J$ with $|S| = k$}
	    \STATE compute two connected components $A,\ A'$ of $G$ induced by the 
			   	colored edges of $G$ with color $i \in S,$ and $i \in I \backslash S$, resp; \label{a:cC}
		\STATE compute $H_1=\la p_{A}(G)\ra$ and $H_2=\la p_{A'}(G)\ra$; 
        \IF{$ H_1 \boxtimes  H_2  \backsimeq G$} \label{a:iI}
         	 \STATE save $H_1$ as prime factor;
             \STATE $J \gets J\backslash S$;
	    \ENDIF
    \ENDFOR
	\ENDFOR
\renewcommand{\baselinestretch}{1.3}
\small\normalsize
  \end{algorithmic}
\end{algorithm}

\begin{figure}[tp]
  \centering
  \includegraphics[bb= 62 626 485 824, scale=0.7]{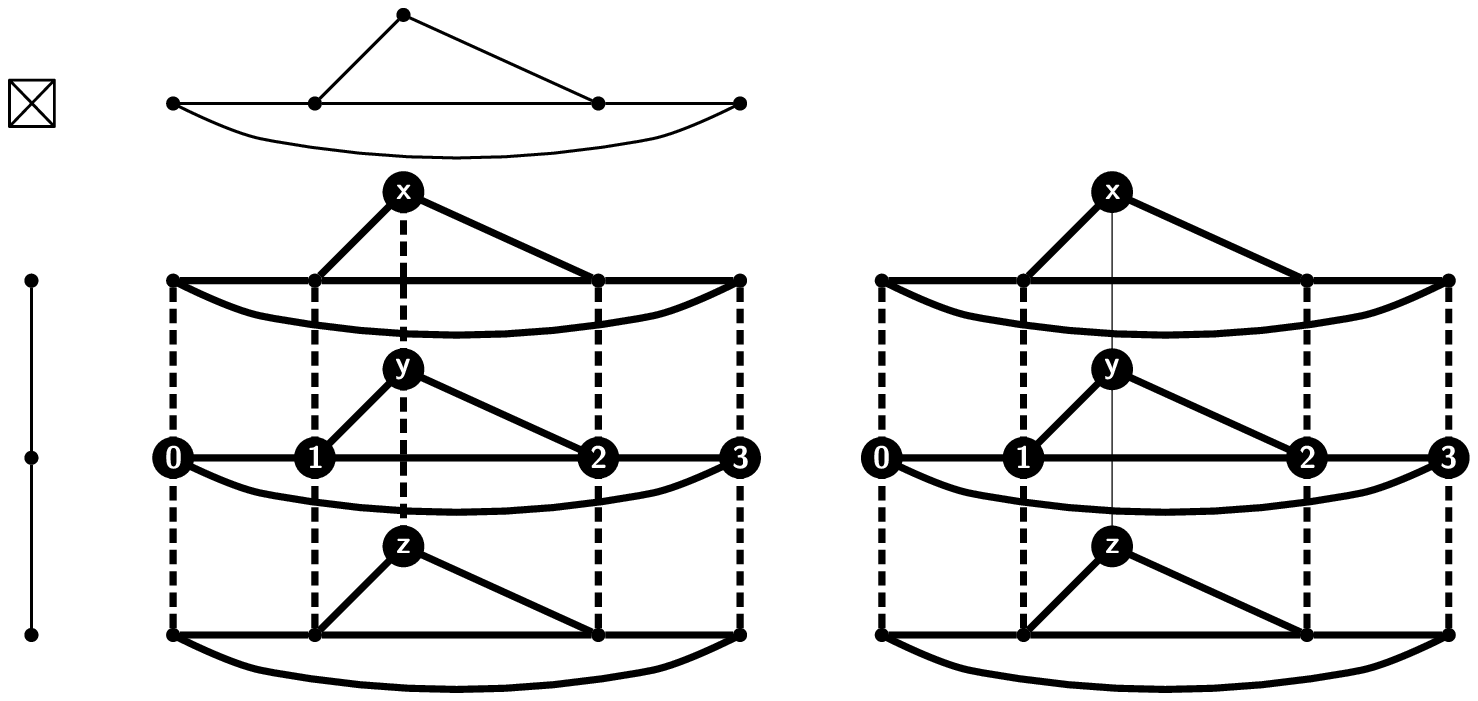}
  \caption[Cartesian skeleton after Algorithm \ref{alg:general}]
	      {Depicted is the colored Cartesian skeleton of the thin strong product graph $G$ after running 
			 Algorithm \ref{alg:general} 
			with different BFS-orderings $\bone_{BFS}$ of the backbone vertices. 
		   The backbone $\bone(G)$ consists of the vertices $0,1,2$ and $3$.\\
		   \textbf{lhs.:} $\bone_{BFS} = 2,1,3,0$.
				 In this case the color-continuation from $N[2]$ to $N[1]$ fails.
				 hence we compute the PFD of the edge-neighborhood
				 $\langle N[2] \cup N[1] \rangle$. Notice that the Cartesian edges $(x,y)$ and $(y,z)$ 
   				 satisfy the \Scond\ in $\langle N[2] \cup N[1] \rangle$ and will be determined as Cartesian.
				In all other steps the color-continuation  works. \\
			\textbf{rhs.:}	$\bone_{BFS} = 3,0,2,1$.
					 In all cases ($N[3]$ to $N[0]$, $N[3]$ to $N[2]$, $N[0]$ to $N[1]$)
					 the color-continuation works. However, after running the first while-loop there are missing 
						Cartesian edges $(x,y)$ and $(y,z)$
					that do not satisfy the \Scond\ in any of the previously used subproducts
					$N[3]$, $N[0]$, $N[2]$ and $N[1]$. Moreover,  
					the edge-neighborhoods $\langle N[x] \cup N[y] \rangle$ as well as 
					$\langle N[z] \cup N[y] \rangle$ are the product
					 of a path and a $K_3$ and the \Scond\ is violated for the Cartesian edges	
					in its edge-neighborhood.
					These edges will be determined
					in the second while-loop of Algorithm \ref{alg:general} using the respective 
					$N^*$-neighborhoods.}
  \label{fig:Exmpl_gen1}
\end{figure}

\begin{thm}
	Given a thin graph $G$ then Algorithm \ref{alg:general} determines 
	the prime factors of $G$ with respect to the strong product. 
	\label{thm:correctness_gen}
\end{thm}
\begin{proof}
	We have to show that every prime factor $G_i$  of $G$ is returned
	by our algorithm.

	First, the algorithm scans all backbone vertices in their BFS-order stored in $\bone_{BFS}$,
	which can be done, since $G$ is thin and hence,  
	$\langle \bone(G) \rangle$ is connected (Theorem \ref{thm:backbone_coverable}).

	In the \emph{first while-loop} one starts
	with the first neighborhood $N[x]$ with $x$ as first
	vertex in $\bone_{BFS}$, we proceed to cover  
	the graph with neighborhoods $N[y]$ with $y\in \bone_{BFS}$
	and $y\in N[x]$. The following cases can occur:
\begin{enumerate}	
\item	 If the color-continuation does not fail there is nothing to do.
		 Furthermore, we can apply Lemma \ref{lem:connectCartSk-via-S1} 
	     and Lemma \ref{lem:every_partialcolor_on_every_vertex}
	 and conclude  that the determined Cartesian edges
	in $\langle N[x] \rangle$, resp.  in $\langle N[y] \rangle$,
	 i.e., the Cartesian edges that satisfy the \Scond\ in $\langle N[x] \rangle$, resp. in $\la N[y] \ra$,   
	induce a connected  
	subgraph of $\langle N[x] \cup N[y] \rangle$. 

\item	If the color-continuation fails, we check if $\langle N[x] \rangle$ and $\langle N[y] \rangle$ 
	are thin. If both neighborhoods are thin we can use 
	Algorithm \ref{alg:checkHypercube} to get a proper color-continuation
	from $\langle N[x] \rangle$ to $\langle N[y] \rangle$ (Lemma \ref{lem:properCC-via-sprime}).

	Furthermore, since both neighborhoods are thin, for all vertices $v$ in $N[x]$, resp. $N[y]$, holds
	$|S_x(v)|=1$, resp. $|S_y(v)|=1$. Hence all edges in $\langle N[x] \rangle$,
	resp. $\langle N[y] \rangle$, satisfy the \Scond.
	Therefore, by Lemma \ref{cor:connectCartSk-via-S1} 
	the Cartesian edges span $\langle N[x] \rangle$ and $\langle N[y] \rangle$ and thus,
	by the color-continuation property, $\langle N[x] \cup N[y] \rangle$ as well.

\item	If one of the neighborhoods is not thin then we check 
	whether the edge $(x,y)$ is dispensable or not w.r.t. $\langle N[x] \cup N[y] \rangle$.  
	If this edge is indispensable then Corollary \ref{cor:proper_edgeN} 
	implies	that $\langle N[x] \cup N[y] \rangle$ is a proper subproduct.
	Corollary \ref{cor:color-conti_edgeN} implies that 
	then get a proper color-continuation from $\langle N[x] \cup N[y] \rangle$
	to $\la N[y] \ra$.

	Furthermore, Lemma \ref{lem:s1_in_H} implies that $|S_{\langle N[x] \cup N[y] \rangle}(x)|=1$.
	and $|S_{\langle N[x] \cup N[y] \rangle}(y)|=1$. 
	From Lemma \ref{cor:connectCartSk-via-S1}  we can conclude that the determined Cartesian edges
	of $\langle N[x] \cup N[y] \rangle$
	induce a connected  subgraph of $\langle N[x] \cup N[y] \rangle$. 
	
\item	Finally, if $(x,y)$ is dispensable in $\langle N[x] \cup N[y] \rangle$ we can not 
	be assured that $\langle N[x] \cup N[y] \rangle$ is a proper subproduct.
	In this case we factorize $N^*_{x,y}$. Corollary \ref{cor:color-conti_N*}
	implies that we get a proper color-continuation from $N^*_{x,y}$ to $\la N[y] \ra$.

	Furthermore,  Lemma \ref{lem:s1_in_H} 
	implies that $|S_{N^*_{x,y}}(x)|=1$ and $|S_{N^*_{x,y}}(y)|=1$. 
	Moreover, from  Lemma \ref{cor:connectCartSk-via-S1}  follows that all Cartesian edges 
	that satisfy the \Scond\ on $N^*_{x,y}$ induce a connected subgraph of $N^*_{x,y}$.
\end{enumerate}
	
	Clearly, the previous four steps are valid for all consecutive backbone vertices $x,y \in \bone_{BFS}$.
	Therefore, we always get a proper combined coloring of 
	$H =  \la \cup_{w\in W} N[w] \ra$ in Line \ref{line-comb}, since $N[x] \subseteq H$ and hence, 
	we always get a proper color-continuation from  $H$ to $N[y]$. 
	Furthermore, by this and the latter arguments in item $1.$--$4.$ 
	concerning induced connected subgraphs we can furthermore conclude that all determined
	Cartesian edges induce a connected subgraph of $H = \la \cup_{w\in \bone(G)} N[w] \ra$.
		The first while-loop will terminate since $\bone_{BFS}$ is finite.
		
    In all previous steps vertices $x$ are marked as "checked" if 
	there is a used subproduct $K$ such that $|S_K(x)| =1$.
	Edges are marked as "checked" if they satisfy the \Scond. 
	Note, after the first while-loop has terminated either 
	edges have been identified as Cartesian 
	or if they have not been determined as Cartesian but satisfy the \Scond\  
	they are at least connected to Cartesian edges that satisfy the 
	\Scond, which follows from 	Lemma \ref{lem:every_partialcolor_on_every_vertex}. 
	This implies that all edges that are marked as "checked" are connected
	 to Cartesian edges that satisfy the \Scond. 
	Moreover, notice that 
	$H = \la \cup_{w\in \bone(G)} N[w] \ra = G$, since $\bone(G)$ is a  
	dominating set.
	
	In the \emph{second while-loop} all vertices that are not marked 
	as "checked", i.e., $|S_K(x)| > 1$ for all used subproducts $K$,
	are treated. For all those vertices the $N^*$-neighborhoods $N^*_{x,y}$ 
	are decomposed and colored. Lemma \ref{lem:N*} implies that $|S_{N^*_{x,y}}(x)|=1$ 
	and $|S_{N^*_{x,y}}(y)|=1$. Hence all Cartesian edges 
	containing vertex $x$ or $y$  satisfy the \Scond\ in $N^*_{x,y}$. 
	Lemma \ref{lem:every_partialcolor_on_every_vertex} implies
	that each color of every factor of $N^*_{x,y}$ 
	is represented on edges containing
	vertex $x$, resp., $y$. Lemma \ref{lem:connectCartSk-via-S1} implies 
	that all Cartesian edges that satisfy the  \Scond\  
	in 	$N^*_{x,y}$ induce a connected subgraph of Lemma $N^*_{x,y}$.

	It remains to show that we get always 
	a proper color-continuation. Since $|S_K(x)| > 1$ for all used subproducts $K$,
	we can conclude in particular that $|S_x(x)| > 1$. 
	Therefore, one can apply
	Lemma \ref{lem:edges_dont_sat_Scond} and conclude that 
	there exists a vertex $z\in \bone(G)$ s.t.\ $z\in N[x]\cap N[y]$ and 
	hence $\la N[z]\ra \subseteq N^*_{x,y}$. This neighborhood $\la N[z]\ra$ 
	was already colored in one of the previous steps since $z\in \bone(G)$. 
	Lemma \ref{lem:s1_in_H} implies that $|S_{N^*_{x,y}}(z)|=1$ and thus
	each color of each factor of $N^*_{x,y}$ is represented on edges containing
	vertex $z$ and all those edges can be determined as Cartesian via the \Scond.
	We get a proper color-continuation from	the already colored subgraph $H$ to $N^*_{x,y}$
	since $N[z] \subseteq H$  and $N[z] \subseteq N^*_{x,y}$,
	which follows from Lemma \ref{lem:color-conti_NtoH}
	and Corollary \ref{cor:color-conti_N*}.

	The second while-loop will terminate since $V(H)$ is finite and
	$|S_{N^*_{x,y}}(x)|=1$ for all $x\in V(H)$.

	As argued before, all edges that satisfy the \Scond, 
	which are \emph{all} edges of $G$ after the second while-loop
	has terminated, are connected to Cartesian edges that satisfy the \Scond. 
	Moreover, all vertices have been marked as "checked". Hence,
	for all vertices holds $|S_K(x)|=1$ for some used subproduct $K$. 
	Since we always got
	a proper combined coloring and hence, a proper
	color-continuation,  we can apply Lemma \ref{lem:every_partialcolor_on_every_vertex}, 
	and conclude that the set of determined Cartesian edges induce
	a connected \emph{spanning} subgraph $G$. Moreover, by the color-continuation
	property we can infer that the final number of colors on
	$G$ is at most the number of colors that were used in the first neighborhood.
	This number is at most $\log\Delta$, since every product
	of $k$ non-trivial factors must have at least $2^k$ vertices.
    Let's say we have $l$ colors. As shown before, all vertices are "checked"
	and thus we can conclude from  Lemma \ref{lem:every_partialcolor_on_every_vertex}
	and the color-continuation	property 
	that each vertex $x\in V(G)$ is incident to an edge with color $c$ for all $c\in \{1,\dots,l\}$.
	Thus, we end with a combined coloring $F_G$ on
	$G$ where the domain of $F_G$ consists of all edges that were
	determined as Cartesian in the previously used 
	subproducts.

	It remains to verify which of the possible factors are
	prime factors of $G$. This task is done by using  Algorithm \ref{alg:iso-test}.
	Clearly, for some subset $S\subset J$, $S$ will contain
	all colors that occur in a particular $G_i$-fiber $G_i^a$
	which contains vertex $a$. Together with the latter arguments
	we can conclude that the set of $S$-colored edges in $G_i^a$
	spans $G_i^a$.
    Since the global PFD induces a local decomposition, even if the 
	used subproducts are not thin, every layer that satisfies the \Scond\  
	in a used subproduct with respect to a local prime factor
	is a subset of a layer with respect to a global prime factor. Thus, we
	never identify colors that occur in copies of different global
	prime factors. In other words, the coloring $F_G$ is a refinement of the
	product coloring of the global PFD, i.e., it might happen that there are 
    more colors than prime factors of $G$.
	This guarantees that a connected component of the graph induced
	by all edges with a color in $S$ induces a graph that is isomorphic
	to $G_i$. The same arguments show that the colors that are not in $S$ lead
	to the appropriate cofactor $H_2$. Thus $G_i$ will be recognized.	
\end{proof}

\begin{rem}
	Algorithm \ref{alg:general} is a generalization
	of the results provided in \cite{HIKS-08, HIKS-09}.
	Hence, it computes the PFD of NICE \cite{HIKS-08} 
	and locally unrefined \cite{HIKS-09} thin graphs.
	Moreover, even if we do not claim
	that the given graph $G$ is thin one can
	compute the PFD of arbitrary graphs $G$ as follows:
	We apply Algorithm \ref{alg:general}
	on $G/S$. The prime factors of $G$ can be constructed
	by using the information of the prime factors of $G/S$
	and application of Lemma 5.40 provided in \cite{IMKL-00}.
\end{rem}
	
In the last	part of this section, we show that Algorithm 
\ref{alg:general} computes the PFD with respect to the 
strong product of any connected thin graph $G$  
in $O(|V|\cdot \Delta^6)$ time. Clearly, this approach is 
not as fast as the approach of Hammack and Imrich, 
see Lemma \ref{lem:complexity_global2a}, but it can 
easily be applied for the recognition of approximate products.

\begin{thm}
	Given a thin graph $G=(V,E)$ with bounded maximum degree
    $\Delta$, then Algorithm \ref{alg:general}
	determines the prime factors of $G$ 
	with respect to the strong product in $O(|V|\cdot \Delta^7)$ time.
\end{thm}
\begin{proof}
	For determining the backbone $\bone(G)$ we have to check for a particular vertex $v\in V(G)$
	whether there is a vertex $w\in N[v]$ with $N[w]\cap N[v] = N[v]$. This
	can be done in $O(\Delta^2)$ time for a particular vertex $w$ in $N[v]$. 
	Since this must be done for all vertices in $N[v]$ we end in time-complexity $O(\Delta^3)$.
	 This step must be repeated for all $|V|$ vertices of $G$. Hence, the time complexity
	for determining $\bone(G)$ is $O(|V| \cdot \Delta^3)$.
    Computing $\bone_{BFS}$ via the breadth-first search takes 
	$O(|V|+|E|)$ time. Since the number of edges is bounded 
	by $|V|\cdot \Delta$ we can conclude that this task needs $O(|V|\cdot \Delta)$ time.
	
	We consider now the Line \ref{b:Part1-start} -- \ref{b:Part1-end} of the algorithm.
	The while-loop runs at most $|V|$ times. Computing $H$ in Line \ref{b:H1}, i.e., adding 
	a neighborhood to $H$, can be done in linear time in 	
	the number of edges of this neighborhood, that is in $O(\Delta^2)$ time.
	The for-loop runs at most $\Delta$ times.
	Each neighborhood has at most $\Delta+1$ vertices 
	and hence at most $(\Delta+1)\cdot \Delta$ edges. 
    The PFD of $\la N[y]\ra$ can be computed in  
	$O(\max(\Delta^2\Delta log(\Delta),\Delta^4))= O(\Delta^4)$ time, 
	see Lemma \ref{lem:complexity_global2a}
	The computation of the combined coloring
	of $H$ and $\langle N[y] \rangle$ can be done in constant time. 
	For checking if the color-continuation is valid
	one has to check at most for all edges of 
	$\langle N[v_i] \rangle$ if a respective colored
	edge was also colored in $H$, which can be done in 
	$O(\Delta^2)$ time.	

	\noindent
	Algorithm \ref{alg:checkHypercube} computes the combined coloring
	of $H$ and $\la N[v_i] \ra$ in  $O(\Delta^2)$ time.
	To see this, notice that 
	\begin{enumerate}
	\item the computation of the coordinates of the product colored
		 neighborhood $\la N[v] \ra$ can be done via a breadth-first search
   	     in $\la N[v] \ra$ in $O(|N[v]| + |E(\la N[v] \ra)|)=
	     O(\Delta + \Delta^2)=O(\Delta^2)$ time.

	\item by the color-continuation property 
    	$H$ can have at most as many colors as there are colors for the first neighborhood
	     $\la N[v_1] \ra$. This number is at most $\log(\Delta)$, 
   	because every product of $k$ non-trivial factors must have at least $2^k$ vertices.
	Thus the for-loop is repeated at most $\log(\Delta)$ times.
	All tasks in between the for-loop can be done in $O(\Delta)$ time
	and hence the for-loop takes $O(\log(\Delta)\cdot \Delta)$ time.

	\item the computation the combined color can be done 
	linear in the number of edges of $\la N[v_i]\ra$ and thus 
	in $O(\Delta^2)$ time.
	\end{enumerate}

	\noindent
	It follows  that all "if" and "else" conditions are bounded
	by the complexity of the PFD of the largest subgraph	
	that is used and therefore by the complexity of 
	the PFD of $N^*_{x,y}$.

	  Each $N^*$-neighborhood has at most $1+\Delta\cdot(\Delta-1)$ vertices. 
		Therefore, the number of edges in each $N^*$-neighborhood is bounded
		by $(1+\Delta\cdot(\Delta-1)) \cdot \Delta$.
		   By Lemma \ref{lem:complexity_global2a} the computation of the PFD of each $N^*$ 
	   and hence, the assignment to an edge of
	   being Cartesian is bounded 
		by $O(\max(\Delta^3 \Delta^2 log(\Delta^2) , \Delta^6)) = O(\Delta^6).$

		Since the while-loop (Line \ref{b:Part1-start}) runs at most $|V|$ times, 
		the for-loop (Line \ref{forLoop}) at most $\Delta$ times and the 
    	the time complexity for the PFD of the largest subgraph 
		is $O(\Delta^6)$, we end in an overall time complexity
		$O(|V|\Delta^7)$ for the first part (Line \ref{b:Part1-start} -- \ref{b:Part1-end})
		of the algorithm.

	Using the same arguments, one shows that the time complexity of 
	the second while-loop is  $O(|V|\cdot\Delta^6)$.
	The last for-loop (Line \ref{b:lastFor1}--\ref{b:lastFor2})
	needs $O(|E|) = O(|V|\cdot \Delta)$ time.
	
	Finally, we have to consider Line \ref{b:iso} and therefore,
	the complexity of Algorithm \ref{alg:iso-test}.
		We observe that the size of $I$ is the number of used colors. As in the 
		proof of Theorem \ref{thm:correctness_gen},  
		we can conclude that this number is 
		bounded by $\log(\Delta)$. Hence, we also have at most $\Delta$ sets $S$, 
		i.e., color combinations, to consider. 
		In Line \ref{a:cC} of Algorithm \ref{alg:iso-test} 
		we have to find connected components of graphs and in Line \ref{a:iI}
		 of Algorithm \ref{alg:iso-test} 
		we have to perform an isomorphism test for a fixed bijection.
		Both tasks take linear time in the number of edges of the graph and 
		hence $O(|V|\cdot \Delta)$ time.

	Considering all steps of Algorithm \ref{alg:general}
	we end in an overall time complexity $O(|V|\cdot\Delta^7)$.
\end{proof}

\section{Approximate Products} 

Finally, we show in this section, how Algorithm \ref{alg:general} can be modified and 
be used to recognize approximate products.
For a formal definition of approximate graph products
we begin with the definition of the distance between
two graphs. We say the \emph{distance} $d(G,H)$ between
two graphs $G$ and $H$ is the smallest integer $k$
such that $G$ and $H$ have representations $G'$, $H'$
for which the sum of the symmetric differences between
the vertex sets of the two graphs and between their edge sets is
at most $k$. That is, if
$$|V(G')\, \triangle\, V(H')|+|E(G')\, \triangle\, E(H')| \leq k.$$
A graph $G$ is a \emph{$k$-approximate graph product} if there is a
product $H$ such that $$d(G,H) \leq k.$$
As shown in \cite{HIKS-08} $k$-approximate  graph products
can be recognized in polynomial time.

\begin{lem}[\cite{HIKS-08}]
\label{lem:k-approx}
For fixed $k$ all strong and Cartesian
$k$-approximate  graph products
can be recognized in polynomial time.
\end{lem}

Without the restriction on $k$ the problem of finding a product of
closest distance to a given graph $G$ is NP-complete for the
Cartesian product. This has been shown by Feigenbaum and Haddad
\cite{FEHA-89}. 
We conjecture that this also holds for the strong product.
Moreover, we do not claim that the new algorithm for
the recognition of approximate products 
finds an optimal solution in general, i.e., a product 
that has closest distance to the input graph. 
However, the given 
algorithm can be used to derive a suggestion of the product
structure of given graphs and hence, of the structure of
the global factors.
For a more detailed discussion on how much perturbation is 
allowed such that the original factors or at least large 
factorizable subgraphs can still be recognized see Chapter 7 in 
\cite{DissHellmuth}.

\begin{figure}[htbp]
  \centering
  \includegraphics[bb= 66 453 503 758, scale=0.6]{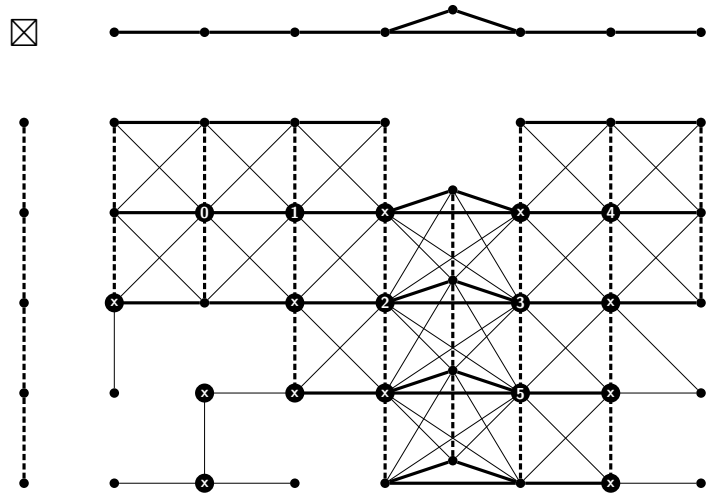}
  \caption[Example: Approximate product 1]{An approximate product $G$ of the product of a path and 
			 a path containing a triangle. 
			 The resulting colored graph after application of  the modified Algorithm \ref{alg:general}
			is highlighted with thick and dashed edges. 
			We set $P=1$, i.e., we do not use prime subproducts and hence
			only the vertices $0,1,\dots,5$ are used. Taking out one maximal component of each color 
			 would lead to appropriate approximate factors of $G$.}
  \label{fig:approxG}
\end{figure}

Let us start to explain this approach by an illustrating example.  
Consider the graph $G$ of Figure \ref{fig:approxG}. 
It approximates $P_5\boxtimes P_7^T$, where $P_7^T$ denotes 
a path that contains a triangle. 
Suppose we are unaware of this fact. Clearly, if $G$ is non-prime, then 
every subproduct is also non-prime.  We factorize 
every suitable subproduct of backbone vertices 
(1-neighborhood, edge-neighborhood, $N^*$-neighborhood) 
that is non-prime
and try to use the information to find a product
that is either identical to $G$ or approximates it.
The backbone $\bone(G)$ is a connected dominating set and 
consists of  the vertices $0,1,\dots,5$ and all vertices marked with "x".
The induced neighborhood of all "x" marked vertices is prime.
We do not use those neighborhoods, but  
the ones of the vertices $0,1, \ldots, 5$, factorize their 
neighborhoods and consider the Cartesian edges that satisfy the \Scond\ 
in the factorizations. 
There are two factors for every such neighborhood
and thus, two colors for the Cartesian edges in every neighborhood.
If two neighborhoods have a Cartesian edge that satisfy the 
\Scond\ in common, we identify their colors.
Notice that the color-continuation fails if we go
from $\la N[2] \ra$ to $\la N[3] \ra$. 
Since the edge $(2,3)$ is indispensable in $\la N[2] \cup N[3] \ra$ and
moreover, $\la N[2] \cup N[3] \ra$ is not prime, one 
factorizes this edge-neighborhood and get a proper
color-continuation. In this way, we end up with two colors altogether, one
for the horizontal Cartesian edges and one for the vertical ones. If
$G$ is a product, then the edges of  the same color span a subgraph
with isomorphic components, that are either isomorphic 
to one and the same factor or that span isomorphic layers 
of one and the same factor.
Clearly, the components are not isomorphic in our example. But, under
the assumption that $G$ is an approximate graph product, we take one
component for each color. In this example, it would be useful 
to take a component of maximal size, say the one consisting of
the horizontal thick-lined edges through vertex $2$, and the vertical dashed-lined 
edges through vertex $3$. These components are isomorphic to the original factors
$P_5$ and $P_7^T$. It is now easily seen that $G$ can be obtained 
from $P_5\boxtimes P_7^T$ by the deletion of edges.
Other examples of recognized approximate products are shown in
Figure \ref{fig:ExmplMoebius} and \ref{fig:approxG2}.

As mentioned, Algorithm \ref{alg:general} has to be modified
for the recognition of approximate products $G$.
We summarize the modifications we apply: 
\begin{enumerate}
\setlength{\parskip}{0pt} 
\setlength{\parsep}{0pt} 
	\item[M1.] $G/S$ is not computed. Hence, we do not claim that the given (disturbed) product is thin.
	\item[M2.] Item M1 and Theorem \ref{thm:backbone_coverable} imply
		  that we cannot assume that the backbone is 
		  connected. Hence we only compute a BFS-ordering
		  on connected components induced by backbone vertices.
	\item[M3.] We only use those subproducts (1-neighborhoods, edge-neighborhood, $N^*$-neighborhood)
		  that have more than $P\geq 1$ prime factors, where $P$ is a fixed integer.
	\item[M4.] We do not apply the isomorphism test (line \ref{isomorph-test-end}).
	\item[M5.] After coloring the graph, we take one minimal, maximal, or  arbitrary connected component 
		  of each color. The choice of this component depends on the problem
		  one wants to be solved.
\end{enumerate}

\begin{figure}[htbp]
  \centering
  \includegraphics[bb= 180 533 368 710, scale =0.9]{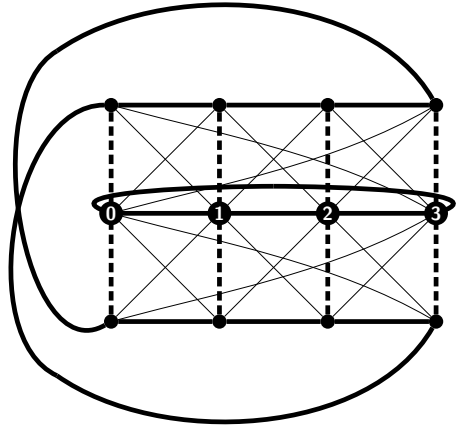}
  \caption[Approximate (twisted) product]{Shown is a prime graph $G$ with $\bone(G) = \{0,1,2,3\}$.
		   This kind of graph is also known as twisted product or graph bundle, see e.g. \cite{IPZ-97, Zer00}. 
	   	   In this example, each PFD of 1-neighborhoods leads to two factors. 
		   Notice that $G$ can be considered as an approximate product of 
		   a path $P_3$ and a cycle $C_4$.
		   After application of the modified Algorithm \ref{alg:general} with $P=1$
		   we end with the given coloring (thick and dashed lines). 
			   Taking one minimal component of each color 
		   	 would lead to appropriate approximate factors of $G$.} 
  \label{fig:ExmplMoebius}
\end{figure}

\begin{figure}[htb]
  \centering
  \includegraphics[bb= 71 452 571 701, scale=0.6]{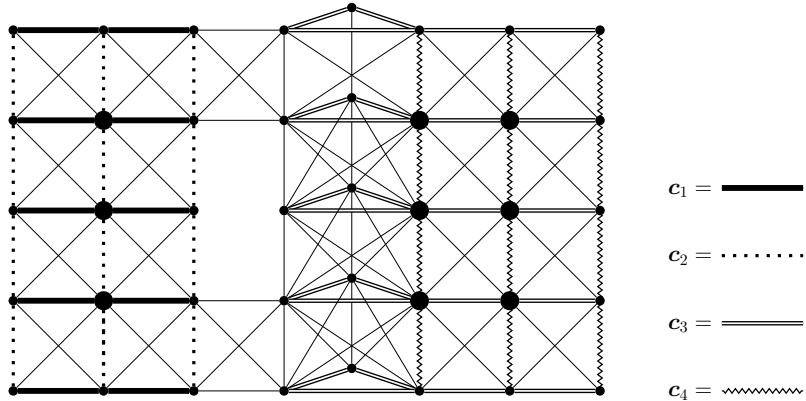}
  \caption[Example: Approximate product 2]{ An approximate product $G$ of the prime factors shown in Figure 
			  \ref{fig:approxG}. In this example $G$ is not thin.
				Obviously, this graph seems to be less disturbed than the one 
				in Figure \ref{fig:approxG}. The thick vertices indicate
				the backbone vertices with more then $P=1$ prime factors.
				Application of  the modified Algorithm \ref{alg:general} on $G$
			   (without computing $G/S$), choosing $P=1$ and using only the thick 
				backbone vertices leads to a coloring with the four colors $c_1,c_2,c_3$ and $c_4$.
				 This is due to the fact that the color-continuation
				fails, which would not be the case if we would allow to use also
				prime regions. }
  \label{fig:approxG2}
\end{figure}

First, the quotient graph $G/S$ will not be computed, since
the computation of $G/S$  of an approximate product graph $G$ 
may result in  a thin graph where a lot of structural information 
has been lost.

Moreover, deleting or adding edges in a product graph $H$, resulting in 
a disturbed product graph $G$, usually makes the graph prime 
and also the neighborhoods $\langle N^{G}[v] \rangle$ 
that are different from $\langle N^{H} [v] \rangle$ and
hence, the subproducts (edge-neighborhood, $N^*$-neighborhood) 
that contain $\langle N^{G}[v] \rangle$. 
In Algorithm \ref{alg:general}, we therefore only use those
subproducts of backbone vertices that are at least not prime, i.e.,  
one restricts the set of allowed backbone vertices to those
where the respective  subproducts have more than $P\geq 1$ 
prime factors and thereby limiting the number of allowed subproducts.
Hence, no prime regions or  subproducts that have less or equal than $P$ 
prime factors are used. Therefore,  
 one does not merge colors of different locally 
determined fibers to only $P$ colors, after the computation of a 
combined coloring.

The isomorphism test (line \ref{isomorph-test-end})
in  Algorithm \ref{alg:general} will not be applied.
Thus, in prime graphs $G$ one does not merge colors if  
the product of the corresponding approximate prime factors
is not isomorphic to $G$.

After coloring the graph, one takes out one
component of each color to determine the (approximate) factors. 
For many kinds of approximate products the connected components
of graphs induced by the edges in one component of each color
will not be isomorphic. In the example in Figure \ref{fig:approxG}, 
where the approximate product was obtained by deleting  edges,
it is easy to see that one should take the maximal connected component
of each color. 

Clearly, this approach needs non-prime subproducts. 
If most of the subgraphs in an approximate product 
$G$ are prime, one would not expect to 
obtain a product coloring of $G$, that can be used to 
recognize the original factors, but that
can be used e.g. for determining maximal factorizable subgraphs
or maximal subgraphs of fibers, see Chapter 7 in \cite{DissHellmuth}. 
Hence, this approach may provide a basis for the development 
of further heuristics for the recognition of approximate products.

\section*{Acknowledgement}

I want to thank 
Peter F. Stadler, Wilfried Imrich and Werner Kl\"ockl for all the outstanding 
and fruitful discussions!
I also thank the anonymous referees for important comments and suggestions.  

\bibliographystyle{plain}
\bibliography{biblio}

\end{document}